\documentclass[12pt,reqno]{amsart}      
\usepackage{amssymb,amsthm}             
\usepackage[dvips]{hyperref}

\usepackage{eucal}                     
\numberwithin{equation}{section}        

\renewcommand{\Re}{{\mathbb R}}         
\newcommand{\la}{\langle}               
\newcommand{\ra}{\rangle}               
\newcommand{\half}{\frac{1}{2}}         
\newcommand{\tr}{{\text{\rm tr}}}	
\newcommand{\ptr}{{\text{\rm tr}_\pme}}	

\newcommand{\pme}{{\tilde g}}     
\newcommand{\pK}{{\tilde K}}      
\newcommand{\pLapse}{\tilde N}    
\newcommand{\pShift}{{\tilde X}}  

\newcommand{\Vol}{\text{Vol}}

\newcommand{\rme}{g}              
\newcommand{\rKtr}{\Sigma}        
\newcommand{\rLapse}{N}           
\newcommand{\rShift}{X}           
\newcommand{\rGamma}{\Gamma}      
\newcommand{\rnabla}{\nabla}      
\newcommand{\rDelta}{\Delta}      
\newcommand{\rR}{R}               

\newcommand{\nLapse}{\omega}
\newcommand{\bme}{\gamma}

\newcommand{\bGamma}{\Gamma[\bme]}
\newcommand{\bnabla}{\nabla[\bme]}
\newcommand{\bR}{R[\bme]}

\newcommand{\bDelta}{\Delta[\bme]}

\newcommand{\SO}{\text{\rm SO}}         

\newcommand{\Ric}{\text{Ric}}
\newcommand{\Scal}{\text{Scal}}
\newcommand{\Riem}{\text{Riem}}

\newcommand{\BB}{\CMcal B} 	
\newcommand{\VV}{\CMcal V} 	
\newcommand{\NN}{\CMcal N} 	
\newcommand{\Slice}{\CMcal S}	
\newcommand{\UU}{\CMcal U}	
\newcommand{\DD}{\CMcal D}	
\newcommand{\MM}{\CMcal M}	
\newcommand{\CC}{\CMcal C}	
\newcommand{\Ein}{\CMcal E}
\newcommand{\SliceCC}{\Slice_{\CC}} 

\newcommand{\EEtot}[1]{E_{#1}}
\newcommand{\EE}[1]{{\CMcal E}_{(#1)}{}}
\newcommand{\GG}[1]{{\Gamma}_{(#1)}{}}
\newcommand{\JJtot}[1]{J_{#1}}
\newcommand{\JJ}[1]{{\CMcal J}_{(#1)}{}}

\newcommand{\Cenerg}{c_{E}}

\newcommand{\Ppara}{{\mathbb P}^{\parallel}}
\newcommand{\Pperp}{{\mathbb P}^{\perp}}
\newcommand{\LL}{\CMcal L}
\newcommand{\LinEin}{\CMcal A}       

\newcommand{\Lie}{\mathcal L}		

\newcommand{\Id}{\mathbf i}		
\newcommand{\Shade}{\CMcal P} 	

\newcommand{\aM}{\bar M} 
\newcommand{\ame}{\bar g}

\newcommand{\eps}{\epsilon}

\newcommand{\TTpara}{\TT \, \parallel} 
\newcommand{\TTperp}{\TT \, \perp} 
\newcommand{\cR}{\overset{\circ}{R}}
\newcommand{\genus}{\text{\rm genus}}
\newcommand{\Tens}{{\mathcal T}}
\newcommand{\Sym}{{\mathcal S}}
\newcommand{\tens}{\otimes} 
 
\renewcommand{\div}{\delta}

\newcommand{\hLapse}{\hat N}            
\newcommand{\tM}{{\tilde M}}            
\newcommand{\del}{\delta}

\newcommand{\FF}{\CMcal F}	

\newcommand{\length}{\text{\em (length)}}
\newcommand{\TT}{\text{\rm TT}}

\newcommand{\lam}{\lambda} 
\newcommand{\lammin}{\lambda_{\text{min}}} 
\newcommand{\lamminmod}{\lammin'}

\newcommand{\calX}{\CMcal X}

\newcommand{\hatV}{\widehat{V}}

\newcommand{\taubar}{\overline{\tau}} 
\theoremstyle{plain}
\newtheorem{thm}{Theorem}[section]
\newtheorem{cor}[thm]{Corollary}
\newtheorem{lemma}[thm]{Lemma}

\newtheorem{definition}[thm]{Definition}
\newtheorem{prop}[thm]{Proposition}
\theoremstyle{remark}
\newtheorem{remark}[thm]{Remark}

\title{Einstein spaces as attractors for the Einstein flow} 
\author[L. Andersson]{Lars Andersson$^1$}
\thanks{$^1$Supported in part by 
the NSF, with grants 
DMS-0407732 and DMS-0707306 to the University of Miami}
\address{Department of Mathematics\\
University of Miami\\
Coral Gables, FL 33124\\
USA \\
\and 
Albert Einstein Institute \\
Am M\"uhlenberg 1\\
D-14476 Golm\\
Germany
}
\email{laan\char'100aei.mpg.de}

\author[V. Moncrief]{Vincent Moncrief$^2$}
\thanks{$^2$Supported in part by the NSF, with grants PHY-0354391 and
PHY-0647331 to Yale University}
\address{Department of Physics and Department of Mathematics\\
Yale University\\
P.O. Box 208120\\
New Haven, CT 06520, USA}
\email{vincent.moncrief@yale.edu}

\begin{document}

\setcounter{tocdepth}{2}
\date{August 5, 2009}


\allowdisplaybreaks[2]

\begin{abstract}
In this paper we prove a global existence theorem, in the direction of
cosmological expansion, for sufficiently small perturbations of a family of 
$n+1$-dimensional, 
spatially compact spacetimes which generalizes the 
$k=-1$ Friedmann--Robertson--Walker vacuum
spacetime. This work extends the result from \cite{AMF}.
The background spacetimes we consider are Lorentz cones over
negative Einstein spaces of dimension $n \geq 3$. 

We use a variant of the constant mean curvature, spatially harmonic (CMCSH) 
gauge introduced in
\cite{AML}. An important difference from the 3+1 dimensional case is that one
may have a nontrivial moduli space of negative Einstein geometries. This
makes it necessary to introduce a time-dependent background metric, which is
used to define the spatially harmonic coordinate system which goes into the
gauge. 

Instead of the Bel-Robinson energy used in \cite{AMF}, we here use an 
expression analogous to the wave equation type of energy introduced in
\cite{AML} for the Einstein equations in CMCSH gauge. 
In order to prove energy estimates, it turns out to be necessary to assume
stability of the Einstein geometry. Further, for our analysis it is necessary
to have a smooth moduli space. Fortunately, all known examples of negative
Einstein geometries satisfy these conditions. 

We give examples of families of Einstein
geometries which have non-trivial moduli spaces. A product construction
allows one to generate new families of examples. 

Our results demonstrate causal geodesic completeness
of the perturbed spacetimes, in the expanding direction, and show that the
scale-free geometry converges towards an element in the moduli space of
Einstein geometries, with a rate of decay depending on the stability
properties of the Einstein geometry. 
\end{abstract}

\maketitle


\section{Introduction} \label{sec:intro} 
Let $M$ be a compact, connected,
orientable manifold of dimension $n \geq 2$, and assume that $M$ admits a
smooth Riemannian Einstein metric $\bme$ with negative Einstein constant. 
After a
trivial rescaling, we can suppose that 
$$
\Ric[\bme] = - \frac{n-1}{n^2} \bme, 
$$
where $\Ric[\bme]$ is the Ricci tensor of $\bme$. With this normalization,
the Lorentz cone spacetime $\aM = (0,\infty)\times M$ with metric 
$$
\ame = -dt\otimes dt + \frac{t^2}{n^2} \bme
$$
is globally hyperbolic and Ricci flat, i.e. a solution of the vacuum Einstein
equations in dimension $n+1$, and admits a homothetic Killing field $Z = t
\partial_t$, such that $\Lie_Z \ame = -2 \ame$.  Such
spacetimes are said, by virtue of this global homothety, to be continuously
self-similar. 

We showed in an earlier paper \cite{AMF}
that in the case $n=3$, Lorentz cone
spacetimes are stable to the future. The main result of \cite{AMF} is 
that for constant mean curvature (CMC) Cauchy data for
the vacuum Einstein equations close to the standard data for a Lorentz cone, the
maximal future Cauchy development is globally foliated by CMC Cauchy
surfaces. Further, the Cauchy development is causally 
geodesically complete to the
future, and the induced spatial metric on the CMC Cauchy surfaces converges,
after a suitable rescaling, to the background metric $\bme$. 
 
The result
in \cite{AMF} required that the background metric $\bme$ satisfy a
nontrivial rigidity condition, namely that it allows no nontrivial, traceless
Codazzi tensors. 
Kapovich \cite[Theorem 2]{kapovich:deform} has proved the existence
of compact hyperbolic spaces with this property, 
see also the discussion in \cite[\S 2.4]{AMF}
The rigidity condition corresponds to the assumption that the moduli space of
flat spacetimes at the Lorentz cone $(\aM, \ame)$ is trivial. The rigidity
condition was later removed by Reiris \cite{reiris-2005-}. 

The argument in
\cite{AMF} and also in \cite{reiris-2005-} relied on Bel-Robinson type energies to
control the (fully non-linear) perturbations, and
is therefore essentially restricted to the 3+1 dimensional case, see however \cite{senovilla}. In this
paper, our aim is to extend the analysis to general dimension. In order to do
this, we introduce a new family of energies for the Einstein equations, which
are not curvature-based. 

\subsection{Lorentz cone spacetimes} 
In the case $n=2$, the Einstein condition implies that 
$(M,\gamma)$ is a hyperbolic surface.
Since other, more far-reaching techniques are available in two
spatial dimensions \cite{mess:const:curv,messnotes,andersson:etal:2+1grav,moncrief:teichmuller}, 
we shall here concentrate on dimensions
$n \geq 3$. 

When $n=3$, $\bme$ is necessarily hyperbolic, i.e. has constant
negative sectional curvature, and indeed hyperbolic metrics provide the most
familiar special cases of negative Einstein metrics in all higher dimensions
as well, but when $n\geq 4$ many examples of non-hyperbolic, negative
Einstein metrics are known to exist \cite{besse:einstein}.

By Mostow rigidity, hyperbolic metrics are unique up to isometry, and trivial
homothetic rescalings, for all dimensions $n \geq 3$ but when $n=2$ there
exists, for each higher genus surface, a finite dimensional Teichm\"uller
space of non-isometric hyperbolic metrics. 
Non-trivial connected
finite dimensional manifolds of negative Einstein metrics can also occur when
$n\geq 4$ but, as we shall see below, these cannot contain a hyperbolic
member, since higher dimensional hyperbolic metrics are always isolated as
Einstein metrics. 

When $\bme$ is hyperbolic, as is necessarily true for $n=2$
and $n=3$, $\ame$ is actually flat and indeed the spacetime 
$(\aM,\ame)$ 
can be regarded as the quotient of the interior of the future light
cone of a point in $n+1$ dimensional Minkowski space by a 
subgroup of the proper orthochronous Lorentz group which fixes that
point. 
When $\bme$ is only Einstein however and does not have constant curvature
the metric $\ame$ is not in general flat. 


If we let $\pme$ and $\pK$ 
denote respectively the first and second fundamental
forms induced on a $t=$constant hypersurface of $(\aM, \ame)$, 
then 
$$
\pme = \frac{t^2}{n^2} \bme, \quad \pK = - t^{-1} \pme
$$
Furthermore, the mean curvature $\tau = \tr_g K = g^{ij} K_{ij}$ is given by
$\tau = - n t^{-1}$ so that the
hypersurfaces of constant $t$ are in fact constant mean curvature (CMC)
slices labelled by the value of their mean curvature which could be used as a
time function. As $\tau$ ranges over $(-\infty,0)$ the spaces
evolve from a zero volume ``big bang'' to an infinite volume limit of
cosmological expansion. 

It is
possible to prove directly, cf. \cite{fischer:moncrief:n+1}, 
that any vacuum spacetime (or, by a
straightforward generalization,   non-vacuum spacetime obeying a suitable
energy condition), which admits a compact, orientable CMC Cauchy hypersurface,
and a non-trivial, proper homothetic Killing field must in fact be a Lorentz
cone spacetime of the
type described above (and so, in particular, devoid of
matter). 
In a suitable time gauge, the reduced Hamiltonian 
takes the value 
$$
H_{\text{reduced}} = |\tau|^n \Vol(M,\pme), 
$$ 
cf. \cite[\S 4.1]{fischer:moncrief:n+1}. The Lorentz cone spacetimes can 
also be characterized uniquely as critical points for the reduced Hamiltonian
\cite[Theorem 3]{fischer:moncrief:n+1}, when the latter is re-expressed in
terms of its natural canonical variables. 

\subsection{Rescaled Einstein equations} 
A closely related charaterization of these spacetimes is that they are the
unique fixed points of what we call the rescaled Einstein equations.  
Noting that the mean curvature has
dimensions $\length^{-1}$, adopting CMC slicing with mean curvature $\tau =
\tr_g K$ as ``time'' and taking the spatial coordinates $(x^i)$ to be
dimensionless we find that the dimensions of the embedding variables of a
Cauchy surface in spacetime are given by 
\begin{align*} 
[\pme_{ij}] &= \length^2, \quad [\pK_{ij}] = \length \\
[\pLapse] &= \length^2, \quad [\pShift^i] = \length
\end{align*} 
We define rescaled, dimensionless 
variables $( \rme, \rKtr, \rLapse, \rShift)$ by
setting 
\begin{align*} 
\rme_{ij} &= \tau^2 \pme_{ij}, \quad \rKtr_{ij} = \tau(\pK_{ij} - \frac{\tau}{n}
\pme_{ij} ) \\ 
\rLapse &= \tau^2 \pLapse, \quad \rShift^i = \tau \pShift^i
\end{align*} 
and rewrite the field equations in terms of these quantities. 
For a treatment using the canonical ADM variables, see 
\cite{fischer:moncrief:n+1}.

When $\tau$ is taken to serve as time all of the conventional ADM equations
--- constraints, evolution equations and gauge fixing equations needed to
enforce the CMC slicing --- become non-autonomous since the mean curvature
$\tau$ appears in all of them. When these are expressed in terms of the
rescaled variables however all of this explicit $\tau$-dependence is scaled
away except for derivatives with respect to $\tau$ which take the
dimensionless form $\tau\partial_\tau$.
One can remove this
final explicit $\tau$-dependence by simply defining a new, 
dimensionless, time coordinate $T$
by 
$$
T = - \ln(\tau/\tau_0)
$$
and reexpressing $\tau\partial_\tau$ as $-\frac{\partial}{\partial T}$. 
Note that the
natural range of $T$ is $\Re$ whereas $\tau$ only ranged over $(-\infty,0)$. 
The
transformed field equations are given explicitly in equations
(\ref{eq:Einrescaled}, \ref{eq:Conrescaled}), 
while the elliptic equations determining the rescaled lapse
function and shift vetor field are given in 
(\ref{eq:defining-resc}). 
The harmonic spatial coordinate condition that
we shall impose later will not disturb the autonomous character of the field
equations but one should note that the inclusions of a cosmological constant
or non-scale-invariant matter sources would disturb this
character. 

While autonomous field equations are not strictly essential for what we wish
to do it is convenient to begin with the simplest cases and to deal with
generalizations later. 

The rescaled variables for the Lorentz cone 
spacetimes defined above are given by 
\begin{align*} 
\rme &= \bme, 
\quad \rKtr = 0 \\
\rLapse &= n, \quad \rShift = 0
\end{align*} 
The time independence of these quantities
shows directly that they are indeed fixed points of the rescaled
equations. By 
\cite[Theorem 2]{fischer:moncrief:n+1} 
they are the only fixed points
of this system and, as we have mentioned, are the only solutions admitting a
non-trivial proper and globally defined homothetic Killing field (namely $Z =
\partial_T$). 

\subsection{Linearized analysis} 
In 
\cite[\S 3]{fischer:moncrief:n+1}, the linearized equations for perturbations
about an arbitrary fixed point were studied. The
transverse-traceless ($\TT$) perturbations 
can be naturally
decomposed in terms of the $\TT$ eigentensors of the operator $\LL$
defined by 
$$
\LL h_{ab}  = - \Delta h_{ab} -  2 R_{acbd} h^{cd}
$$
where $\Delta = \bme^{cd} \nabla_c \nabla_d$, with $\nabla$ the
covariant derivative defined with respect to $\bme$, and $R_{abcd}$ the
Riemann tensor of $\bme$. The above operator is closely related to 
the one defined in  \cite[equation (3.12)]{fischer:moncrief:n+1}, and also  
to the
Lichnerowicz Laplacian $\Delta_L$, cf. equation (\ref{eq:DeltaL}). 
Note that the 
compatibility of this operator with the $\TT$ character of the eigentensor
depends upon the fact that $\bme$ is Einstein, see section
\ref{sec:stabein} for details.

The eigenvalues $\lambda$ of $\LL$ are all real, since $\LL$ is
self-adjoint with respect to the natural $L^2$ inner product. A separation of
variables argument may be used to analyze the linearized, rescaled, Einstein
equations, cf. section \ref{sec:linearized} below, see also
\cite{fischer:moncrief:n+1}.   
The character of the corresponding solution
depends upon the value of $\lambda$ as follows. 
If 
$\lambda > \frac{(n-1)^2}{4n^2}$, then the characteristic equation has a
complex pair of roots with real part $-(n-1)/2$, and hence there is a
universal exponential rate of decay $-(n-1)/2$, in the time $T$. 
If $0 < \lambda < \frac{(n-1)^2}{4n^2}$ the
characteristic equation has a pair of negative real roots. 
In this case we have an ``anomalous'' rate of decay
depending on $\lambda$. 
In the marginal
case $\lambda = \frac{(n-1)^2}{4n^2}$, the system has a resonance. We avoid
dealing directly with the marginal case by considering a slightly decreased
$\lambda$. 
If $\lambda = 0$, we would have a 
``neutral'' mode which does
not decay. Typically, this situation corresponds to the existence of a
nontrivial moduli space of Einstein metrics containing $\bme$. 

If $\lambda < 0$ were to occur, the characteristic equation would
have a root with positive real part and the corresponding solution would 
grow exponentially rather than decay. It is an open question whether any such
``unstable'' Einstein spaces exist. The above discussion
motivates calling $\bme$ stable if $\LL$ has 
non-negative spectrum. We shall review what
is known below, cf. section \ref{sec:stabein}.

\subsection{Stable Einstein spaces and moduli spaces} \label{sec:stabmod}
When $n=3$ an Einstein metric is necessarily hyperbolic (in the negative case
of interest here) and Mostow rigidity excludes the possibility of deforming
the hyperbolic (hence Einstein) structure. In this sense $n=3$ is the most
``rigid'' dimension -- a manifold either admits no Einstein structure or
precisely one. 

For $n> 3$ Mostow rigidity still applies but now the new possibility arises
of having negative Einstein spaces that are not hyperbolic. Many families of
such (negative) Einstein spaces are known to exist and whenever a
chosen background metric $\bme$ belongs to a smooth (necessarily finite
dimensional modulo gauge degrees of freedom) 
such family the linearized equations will always admit a
corresponding (finite dimensional) space of neutral modes with 
$\lambda = 0$. These represent the tangent space, at the
given background, to the space of self-similar
spacetimes. 
Such smooth families of self-similar spacetimes, determined by the
corresponding families of negative Einstein metrics, are expected to form
``center manifolds'' for the dynamical system defined by the rescaled
Einstein equations and we shall see below that this is in fact the case. 

If there are no obstructions to integrating
an infinitesimal Einstein deformation to a curve of Einstein structures, then 
the moduli space is a manifold, cf. \cite[\S 12.F]{besse:einstein}. We refer
to such moduli spaces as integrable, cf. definition \ref{def:integrable}
below. 
In particular, there are
examples of negative Einstein spaces contained in an integrable 
moduli space, such that $\LL$ has non-negative spectrum.

One such family is
given by negative K\"ahler-Einstein metrics, cf. section \ref{sec:examp-stab}.
Hyperbolic metrics in dimension $n \geq 3$ 
are rigid, in the sense that the moduli space of Einstein
metrics is trivial, and further are strictly stable, 
in the sense that the spectrum of
$\LL$ is positive, see section \ref{sec:stable}. 
If
$n=2$ however one can show that zero is always in the spectrum of $\LL$, 
since every $\TT$ tensor on a higher genus surface is
a traceless Codazzi tensor \cite{andersson:etal:2+1grav}. 
This corresponds precisely to the
presence of a full Teichm\"uller space of self-similar solutions to the
Einstein equations to which the
tangent space at any one corresponds to the space of ``neutral modes''
defined by $\TT$ tensors. This possibility arises precisely for
$n=2$ by virtue of the failure of Mostow rigidity for hyperbolic structures
to hold for surfaces.

\subsection{Stability of Lorentz cone spacetimes} 
The main result of this paper gives a nonlinear stability result for Lorentz
cone metrics, which generalizes the results of \cite{AMF,reiris-2005-} from the 3+1
dimensional to the $n+1$ dimensional case. The result requires that $\bme$ be
stable in the sense discussed in section \ref{sec:stabmod} above, and that
$\bme$ be either rigid, or contained in an 
integrable moduli space of Einstein structures. For data sufficiently close
to the standard data for the Lorentz cone over $\bme$, the rescaled
geometry tends in the expanding direction 
to a limit in the moduli space of $\bme$. It is this
fact which motivates the title of the paper. 
For the case of Ricci flow, stable Einstein spaces also play the role of
attractors, in a sense which is closely related to the one discussed above,
see \cite{dai:etal:kahler}.  

The idea that (stable) Einstein
spaces are attractors for the Einstein flow is motivated by 
the linearized analysis in
\cite{fischer:moncrief:n+1}, see section \ref{sec:linearized} below,
as well as by the fact that the reduced
Hamiltonian $H_{\text{reduced}}$ has positive semi-definite Hessian at the
Lorentz cone data exactly when $\bme$ is stable, 
together with the observation that
$H_{\text{reduced}}$ is monotone decreasing to the future, 
cf. \cite{fischer:moncrief:n+1}.

The main result in this paper is analogous to
the result of \cite{andersson:etal:2+1grav}, where it is
shown that the rescaled geometry of the level sets of the mean curvature time
function converges to a point in Teichm\"uller space, with the Einstein
moduli space playing the same role as Teichm\"uller space.

The work in \cite{AMF} relied on the analysis of the Einstein equations in
CMCSH gauge, i.e. CMC time gauge with spatial harmonic coordinates. 
In the present
work, due to the presence of nontrivial moduli spaces of Einstein background
metrics, it is necessary to use a generalization of the CMCSH gauge, which
allows for a time-dependent background metric for the spatial harmonic
coordinates. 
The time-dependent background metric is determined by a so-called
shadow metric condition, which requires that the difference between the
rescaled metric $\rme$ and the background metric used to define the spatial
harmonic gauge be $L^2$-orthogonal to the deformation space (or premoduli
space, in the terminology of \cite{besse:einstein}). 
The presence of a nontrivial moduli space makes
necessary some rather delicate considerations in order prove the required 
energy 
estimates.

\subsection{Overview of this paper} 
Section \ref{sec:stabein} gives the necessary background material on negative
Einstein spaces, and introduces the notion of stability. In section
\ref{sec:examp-stab} we collect some known examples of stable negative
Einstein spaces with integrable moduli space. In section \ref{sec:product} we
show that Cartesian products of stable spaces are stable, and that taking
products of stable spaces with integrable moduli space, yields new spaces
with the same property. This allows us to construct large families of
examples where our results apply. 

Next, in section \ref{sec:rescaled}, we introduce the rescaled system of
Einstein equations which will be studied, together with the shadow metric
condition, that is the gauge condition generalizing the CMCSH condition of \cite{AML}
which we use in the case when the moduli space is nontrivial. Section
\ref{sec:locex} discusses the proof of local wellposedness for the rescaled
Einstein equations with shadow metric gauge, based on the work in
\cite{AML}. 

The linearized, rescaled, Einstein equations are introduced and analyzed in
section \ref{sec:linein-energ}. This section also contains an analysis of the
damped oscillator equation which arises from a separation of variables of
these same linearized equations. The behavior of the solutions of
this equation was studied in \cite{fischer:moncrief:n+1}. However, for the
present purposes, we need an
energy argument which yields the correct
decay estimates. The energy for the damped oscillator equation is analyzed in
section \ref{sec:damped}. In section \ref{sec:energy}, this analysis is used
as the basis for a definition of energies for the full rescaled 
Einstein equations. The energies we use have a lot in common with the energies
used for the local existence proof in CMCSH gauge, cf. \cite{AML}. 
Finally, section \ref{sec:global} gives the statement and proof of the main
results of the paper.  

The appendix \ref{sec:shadow:gauge} provides the proof of Lemma
\ref{lem:shadow-exist-met}. 

\subsection{Preliminaries and notation}  \label{sec:prelim}
Let $M$ be a compact manifold of dimension $n \geq 2$ and let $\MM$ denote the
space of Riemannian metrics on $M$. For $g \in \MM$, we denote by $\Riem,
\Ric, \Scal$, the Riemann and Ricci tensor, and the scalar curvature,
respectively. We shall often use index notation, with lower case latin
indices running over $1,\dots,n$ and greek indices running over $0,\dots,n$. 
The index versions of $\Riem$
and $\Ric$ are 
$R_{abcd}$, and $R_{ab} = R^c{}_{acb}$, respectively. The 
Christoffel symbols of $g$
are denoted by $\Gamma^i_{mn}$. We shall sometimes indicate that a curvature 
tensor or
Christoffel symbol is
defined with respect to a special metric, say $\bme$, by writing,
e.g., $\bR_{abcd}$, $\bGamma_{mn}^i$. 


We shall often work in the context of Sobolev regular metrics. For an integer
$s$, let $H^s$
denote the $L^2$ Sobolev spaces on $M$, defined with respect to some once and
for all given background metric. We use the notation $|| \cdot ||_{H^s}$ for
the $H^s$ norm. For $s > n/2
+1$, denote by $\MM^s$, 
$\Ein^s_{\alpha}$ the spaces
of metrics, 
and Einstein metrics with Einstein constant $\alpha$, of
Sobolev class $s$. 
For an Einstein metric, we may without loss of
generality, by working in harmonic coordinates, assume that it is $C^\infty$ or in fact real analytic, see
\cite{deturck:kazdan:regularity}. 
For most standard situations we
leave it to the reader to fill in the analytical details and drop the Sobolev
index from our notation.

\section{Negative Einstein spaces} \label{sec:stabein}
In this section we shall review some material on negative Einstein spaces
which is needed in the rest of the paper. We emphasize the notion of
stability for Einstein spaces, cf. section \ref{sec:stable}. 
The book \cite{besse:einstein} is a good general reference on Einstein
metrics. In particular, \cite[Chapter 12]{besse:einstein} contains a
discussion of the moduli space of Einstein structures. In
\cite{besse:einstein}, the space of Einstein geometries on a compact manifold
$M$ is studied as the space of metrics of unit volume, modulo
diffeomorphisms, which solve the equation 
$$
\Ric = \frac{1}{n} \left (\int_M \mu_g \Scal \right ) g
$$
The scalar curvature is locally constant on the space of
Einstein geometries, cf. \cite[Corollary 12.52]{besse:einstein}.
In this paper, we are
interested only in connected components of the space of Einstein
geometries with fixed negative Einstein constant. 
Thus, we fix $\alpha < 0$ and consider without loss of generality the space 
$\Ein_\alpha$ of Einstein metrics on $M$ with
Einstein constant $\alpha$, i.e. the space of solutions to the Riemannian 
Einstein equation 
$$
\Ric = \alpha g .
$$
We assume that $\Ein_\alpha$ is non-empty.
The results from \cite{besse:einstein} specialize to the situation considered
here. 
As
in \cite{besse:einstein}, we shall work with the premoduli space $\Ein_\alpha
\cap \Slice$, where $\Slice$ is a slice for the diffeomorphism group. 
Given a metric $\gamma_0 \in \Ein_\alpha$ 
we refer to the connected component of the premoduli space of $\gamma_0$ as
the deformation space of $\gamma_0$. The moduli space of $\gamma_0$ is
the quotient of the premoduli space by the isometry group of $(M,\gamma_0)$,
which in case $\alpha < 0$ is finite, see \cite[\S 12.C]{besse:einstein} for
further discussion. 

\subsection{The linearized Einstein equation} \label{sec:linein}
Denote by $\LinEin$ the 
operator, 
$$
\LinEin h = 2 D (\Ric - \alpha g) h ,
$$
i.e. twice the Frechet derivative of $g \mapsto \Ric - \alpha g$, in the
direction $h$, evaluated at $\bme \in \Ein_\alpha$. 
Let $\LL$ be given by 
\begin{equation}\label{eq:LLdef} 
\LL h = - \Delta h - 2 \cR h
\end{equation} 
where $\Delta = \nabla^a \nabla_a $ and 
\begin{equation}\label{eq:cRdef}
\cR h_{ab} = R_{acbd} h^{cd} 
\end{equation} 
Then
$$
\LinEin h  = \LL h - 2 \div^* \div h - \nabla d (\tr h)
$$
where 
$$
\div h_a = \nabla^b h_{ab}, \quad \div^* \xi_{ab} = - \half (\nabla_a \xi_b +
\nabla_b \xi_a) 
$$
Symmetric two-tensors with
vanishing trace and divergence, i.e. elements of $\ker \div \cap \ker \tr$  
play an important role in analyzing the 
linearized Einstein equation. Such tensors are called
transverse and traceless or $\TT$ tensors. 
The space 
$$
\ker \LinEin \cap \ker \div \cap \ker \tr
$$ 
is the space of infinitesimal Einstein deformations in our setting, cf. 
\cite[Theorem 12.30]{besse:einstein} for the analogous statement in their
setting. 

The operator $\LL$ is a self-adjoint 
elliptic operator and since by assumption $M$ is
compact, $\LL$ has discrete spectrum and finite dimensional kernel. 

We remark that $\LL$ is closely related
to the Lichnerowicz Laplacian $\Delta_L$ defined by 
\begin{equation}\label{eq:DeltaL}
\Delta_L h = \LL h + 2\Ric \circ h 
\end{equation} 
where for two symmetric tensors $u,v$, 
$$
(u \circ v)_{ab} = \half ( u_{ac} v^c{}_b + v_{ac} u^c{}_b ) .
$$
Let $\Delta_H = d \div + \div d$ be the Hodge Laplacian, where $\div$ is the
adjoint of $d$. 
For $\bme \in \Ein_\alpha$ the commutation
formulas 
$$
\div \LL u = (\Delta_H  - 2 \alpha) \div u 
$$
and 
$$
\tr \LL u = (- \Delta  - 2 \alpha ) \tr u
$$
hold. 
The Hodge Laplacian acts on 
on one-forms by 
$$
\Delta_H \xi = (- \Delta + \Ric) \xi 
$$
In particular, $\Delta_H$ has non-negative spectrum. 
It follows from the commutation formulas that $\LL$ maps $\TT$ tensors to
$\TT$ tensors. Further, if $\alpha < 0$, then $\ker \LL \subset \ker \div
\cap \ker \tr$, i.e. a subspace of the space of $\TT$ tensors. 

For a symmetric 2-tensor $h$ we have the decomposition 
\begin{equation}\label{eq:TTsplit} 
h = f g + h^{\TT} + \Lie_Y g , 
\end{equation} 
valid at any metric $g$, where $f$ is a function, $h^{\TT}$ is a $\TT$ tensor
with respect to $g$, and $Y$ is a vector field. 
The equation $\Ric = \alpha g$ 
is covariant which implies that at $\bme \in \Ein_\alpha$, 
$\LinEin ( \Lie_Y \bme) = 
0$, for any vector field $Y$. 
In the rest of this subsection, we evaluate
$\LL$ and $\LinEin$ at $\bme \in \Ein_\alpha$. 
Since $\LinEin h^{\TT} = \LL h^{\TT}$, a calculation shows 
\begin{equation}\label{eq:LinEin}
\LinEin h = \LL h^{\TT} + [(-\Delta - 2\alpha) f] \bme + (2-n) \nabla df 
\end{equation}
Taking the trace, we find that if $\alpha < 0$,
$h \in \ker \LinEin$ only if $f=0$. This shows
\begin{equation}\label{eq:kerLinEin} 
\ker \LinEin = \ker \LL + \{ \Lie_Y \bme \ , \ Y \text{ vector field on $M$} \}
\end{equation} 
Let $\Ppara$ denote the $L^2$-orthogonal projection in the space of symmetric
2-tensors onto the
the finite dimensional kernel $\ker \LL$ and similarly let $\Pperp$ 
be the $L^2$-orthogonal projection in the space of symmetric 2-tensors 
onto the $L^2$-orthogonal complement of $\ker \LL$ in
the space of $\TT$ tensors with respect to $\bme$. 

Given a symmetric 2-tensor $u$, we will
often use the notation 
$u^{\TT}$, $u^{\TTpara}$ and $u^{\TTperp}$ for the $\TT$ part of $u$, in the sense of the decomposition
(\ref{eq:TTsplit}),  and the
projections $\Ppara u$, $\Pperp u$, respectively. In particular, 
a  $\TT$ tensor  $u^{\TT}$ can be decomposed as  $u^{\TT} = u^{\TTpara} + u^{\TTperp}$. 
If it is not clear from the context we will indicate the dependence on the
metric by e.g. $\Ppara_{\bme}$. 
If $\ker \LL = \{0\}$, then 
$\Pperp$ is the projection onto the space of $\TT$ tensors with respect to
$\bme$.  

\subsection{Harmonic coordinates and the slice} \label{sec:HC} 
For $\bme \in \Ein_\alpha$, 
a slice $\Slice_{\bme}$ for the diffeomorphism group, called 
the harmonic slice through $\bme$,  can be defined as follows. 
Let
$\Slice_{\bme}$ be the set of $g \in \MM$ such that the identity map 
$\Id: (M,g) \to (M,\bme)$
is harmonic. 
This condition holds if and only if the tension
field $V$ vanishes, where $V$ is given by 
\begin{equation}\label{eq:Vdef}
V^i[g;\bme] = g^{mn} (\Gamma[g]^i_{mn} - \bGamma^i_{mn} ) 
\end{equation} 
For $\bme \in \Ein_\alpha$ for $\alpha < 0$ and 
$g$ sufficiently close to $\bme$, we have that 
$$
g^{mn} \bR_{imjn} X^i X^j \leq - \lambda^2 |X|^2_\bme
$$
for some $\lambda > 0$. In fact it is sufficient for $\bme$ to have
negative Ricci curvature for this to hold. 
Therefore we have following \cite[\S 5]{AML}, see in particular 
\cite[(5.7)]{AML}, that the operator $P$ defined by 
\begin{equation}\label{eq:Pdef}
PX = D V. \Lie_X \bme ,
\end{equation} 
where the Frechet derivative of $V$ is taken with respect to $g$, 
is an isomorphism at $\bme$, and the same holds, by continuity, 
for the corresponding operator defined
at $g$, for $g$ close to $\bme$. 
Consider a symmetric 2-tensor $h$ decomposed as in (\ref{eq:TTsplit}), with
respect to $g$. We
have for $g \in \Slice_\bme$, 
\begin{equation}\label{eq:DV}
(DV . h)^i = P Y^i + (1-\frac{n}{2}) \nabla^i f - h^{\TT\, mn} (\Gamma[g]^i_{mn} -
\bGamma^i_{mn}) 
\end{equation} 
%
Based on this, is not difficult to apply the implicit function theorem
to show that  $\Slice_{\bme}$ is a 
submanifold of $\MM$ near $\bme$. It follows from the same analysis that
$\Slice_{\bme}$ defined in this way is a slice for the diffeomorphism
group, see also \cite[\S 12.C]{besse:einstein} for discussion. 

On the other hand, for $g \in \MM$, close to a negative Einstein metric $\bme$,
there is a harmonic map $\phi: (M,g) \to (M,\bme)$, with $\phi \in
\DD$ being the unique solution to the harmonic map equation 
$$
\Delta \phi^i + \bGamma^i_{mn} \partial_k \phi^m \partial_l \phi^n g^{kl} = 0
$$
in a neighborhood of the identity map $\Id$. 
This is proved along the lines of 
the above remarks, also using an implicit function theorem argument. Given
$\phi$, the pushforward 
$(\phi^{-1})^* g$ has the property that $\Id :
(M, (\phi^{-1})^* g) \to (M,\bme)$ is harmonic, and also that 
$\Id : (M,g) \to (M,
\phi^* \bme)$ is harmonic. Since $\bme \in \Ein_\alpha$, this holds for
$\phi^* \bme$ too. Thus, if $g$ is close in $\MM$ 
to  $\bme \in \Ein_\alpha$ we may,
after applying a diffeomorphism, 
assume that in fact $g \in \Slice_{\bme}$.
%
If it is clear from the context which metric is used to define the slice, we
will simply denote the slice by $\Slice$.

\subsection{The deformation space} \label{sec:slicedeform} 
Fix $\bme_0 \in \Ein_\alpha$ and for the rest of this section, let $\Slice$
be the slice defined with respect to $\bme_0$. 
\begin{definition} Let $\bme_0 \in \Ein_\alpha$, and let 
$\VV$
be the
  connected component of $\bme_0$ in $\Ein_\alpha$.  
The space $\NN = \VV \cap \Slice$ 
is called the {\bf deformation space} of $\bme_0$. If $\NN = \{\bme_0\}$,
then $\bme_0$ is called {\bf rigid}. 
\end{definition} 
\begin{remark} 
By \cite[Corollary 12.52]{besse:einstein}, the moduli space of Einstein
structures is locally connected. 
In the terminology of \cite[\S 12]{besse:einstein}, the 
deformation space is the connected component of 
  the premoduli space, which contains $\bme_0$.
\end{remark} 
By definition, $\NN$ is the $\bme_0$ component in the space of solutions of 
\begin{equation}\label{eq:RicV} 
\Ric = \alpha g , \quad V = 0
\end{equation} 
By (\ref{eq:kerLinEin}) $h$ solves the linearization of the equation 
$\Ric = \alpha g$ at $\bme_0$ if and only if 
\begin{equation}\label{eq:kerDEin}
h = h^{\TTpara} + \Lie_Y \bme_0,
\end{equation}
with $h^{\TTpara} \in \ker \LL$, where
$\LL$ is the operator defined by (\ref{eq:LLdef}) at $\bme_0$. From the
discussion in section \ref{sec:HC}, we see that the space of solutions to the
linearization of the system (\ref{eq:RicV}) at $\bme_0$ is equivalent to 
$\ker \LL$.  By
\cite[Corollary 12.66]{besse:einstein}, $\bme_0$ is rigid
  if $\ker \LL = 0$.

Next we consider the linearization of the system (\ref{eq:RicV}) away from
$\bme_0$. Thus, let 
$\bme \in \NN$ be close to $\bme_0$. Again, the space of solutions of the
linearized Riemannian Einstein equations has the form (\ref{eq:kerDEin}). 
However, for $\bme \ne
\bme_0$, the term 
$$
h^{\TTpara\, mn} (\Gamma[\bme]_{mn}^i - \Gamma[\bme_0]_{mn}^i) 
$$ 
is nontrivial, and the space of solutions of the linearization of the
system (\ref{eq:RicV}) is of the form 
\begin{equation}
\label{eq:TNN}
h^{\TTpara} + \Lie_{Y^\parallel} \bme ,
\end{equation} 
where $h^{\TTpara} \in \ker \LL$ with $\LL$ defined with respect to $\bme$,
and $Y^\parallel$ solves the equation 
\begin{equation}\label{eq:PYi}
PY^{\parallel\, i} - h^{\TTpara\, mn} (\Gamma[\bme]_{mn}^i - \Gamma[\bme_0]_{mn}^i) = 0 .
\end{equation} 
Here $P$ is the operator given by (\ref{eq:Pdef}), defined with
respect to $\bme$ with background metric $\bme_0$. 
Thus the space of $h$ of the form (\ref{eq:TNN}) is the formal
tangent space of the deformation space $\NN$ at $\bme$. 
\begin{definition} \label{def:integrable} 
If $\NN$ is a manifold near $\bme$, with tangent space given by the formal
tangent space, i.e. 
$$
T_\bme \NN = \{h : h = h^{\TTpara} + \Lie_{Y^\parallel} \bme\},
$$ 
with $Y^\parallel$ a solution to (\ref{eq:PYi}), then
$\NN$ is called {\bf integrable} at $\bme$. 
\end{definition} 
See \cite[\S\S 12.E,F]{besse:einstein} for a discussion of integrability. 
%
\begin{remark}
No example of a compact, negative Einstein space with non-integrable
deformation space is
known. 
\end{remark} 
From now on, we shall consider only integrable deformation spaces.

\subsection{Stable negative Einstein spaces} \label{sec:stable} 
\begin{definition} Let $(M,\bme)$ be a negative 
Einstein space and let $\lammin$ be the
  lowest eigenvalue of $\LL$. We call $(M,\bme)$ {\bf
stable} if $\lammin \geq 0$, while if $\lammin > 0$, $(M,\bme)$ is
  called {\bf strictly stable}. 
\end{definition} 
\begin{remark} 
No example of an unstable, compact, negative Einstein space is known. 
\end{remark} 

The following Lemma shows that stability can be analyzed by looking at the
restriction of $\LL$ to $\TT$ tensors. 
\begin{lemma} \label{lem:nonneg} 
Suppose $(M,\bme)$ is an Einstein space with 
Einstein constant $\alpha < 0$. Then $\LL$ has nonnegative spectrum 
as an operator on
$\TT$ tensors if and only if $(M,\bme)$ is stable. 
In particular, $\ker \LL \subset \ker
\div \cap \ker \tr$ so any element of $\ker \LL$ is a $\TT$
tensor. 
\end{lemma} 
\begin{proof} 
Let $u$ be an eigentensor of $\LL$, 
$$
\LL u = \lambda u
$$
and suppose $\lambda \leq 0$. This implies 
$$
\div \LL u = \lambda \div u 
$$
and hence 
$$
\Delta_H \div u  = (2 \alpha + \lambda) \div u 
$$
i.e. either $2\alpha + \lambda < 0$ is an eigenvalue of $\Delta_H$ or $\div
u = 0$. Since $\Delta_H$ is nonnegative, we have $\div u = 0$. Similarly,
by applying the trace to both sides, we find 
$$
- \Delta \tr u = (2\alpha + \lambda) \tr u
$$
which implies $\tr u = 0$. Hence $u$ is a $\TT$ tensor. This shows that if 
$\LL$
has nonnegative spectrum when restricted to $\TT$ tensors then $\LL$ has
nonnegative spectrum. The converse
is immediate, as is the statement about $\ker \LL$. 
\end{proof}

\begin{prop} \label{prop:uniform} 
\begin{enumerate}
\item \label{point:unif-1}
If $\bme_0 \in \Ein_\alpha$ is strictly stable
then $\NN = \{ \bme_0 \}$. 
$\bme_0$ is isolated in $\Ein_\alpha \cap \Slice$. 
In particular, there is a neighborhood
$\UU$ of $\bme_0$ in the space of metrics such that  
$\UU \cap \Ein_\alpha \cap \Slice = \{\bme_0\}$. 
\item \label{point:unif-2} 
If $\bme_0 \in \Ein_\alpha$ 
has integrable, nontrivial deformation
space, then 
there is a neighborhood $\UU$ of $\bme_0$ in the space of metrics, such
that $\UU \cap \Ein_\alpha \cap \Slice = \UU \cap \NN$.
\end{enumerate} 
\end{prop} 
\begin{proof} 
Point \ref{point:unif-1} is a special case of \cite[Corollary
  12.66]{besse:einstein}. 
Point \ref{point:unif-2} is the fact that the moduli space and in
particular the premoduli space is locally arcwise connected,
cf. \cite[Corollary 12.52]{besse:einstein}. 
\end{proof}

\subsection{Examples} \label{sec:examp-stab}
We will now discuss some general conditions which imply that a negative
Einstein space is stable or strictly stable, and give examples of such spaces
as well as of stable negative Einstein spaces with nontrivial deformation
space. 
\subsubsection{Strictly stable spaces}
A compact Einstein space of dimension $\geq 3$, 
with negative sectional
curvature is strictly stable, cf. \cite[\S 12 H]{besse:einstein}. 
Thus,
in particular, compact rank one symmetric spaces of noncompact type provide
examples of strictly stable negative Einstein spaces. These include, among
others, compact
hyperbolic and  complex hyperbolic spaces. 
More generally, locally symmetric spaces of non compact type, which have
no local 2-dimensional factor are rigid \cite[Prop. 12.74]{besse:einstein}
and hence also strictly stable. 
We mention also the work of Fischer and Moncrief
\cite{fischer:moncrief:n+1} which gives a condition on the Weyl tensor which
implies strict stability. 

\begin{remark} 
In certain cases, a type of uniformization theorem is valid.  
If $M$ carries a hyperbolic 
metric, then by work of Besson et al. \cite{BCG2}, an Einstein
metric on $M$ is hyperbolic. Further, in dimension 4, if $M$ carries a
complex hyperbolic metric, then any Einstein metric on $M$ is complex
hyperbolic by work of LeBrun \cite{lebrun:rigid}. 
\end{remark} 

\subsubsection{Nontrivial deformation spaces} 
In the two
dimensional case, negative curvature does not imply strict stability. 
In particular, from standard results in
Teichmuller theory, in the case $n=2$, 
a Riemann surface $M$ of 
$\genus(M) > 1$, with the hyperbolic metric is a stable negative Einstein
space with nontrivial, integrable deformation space, namely the Teichmuller
space. In this case $\dim \ker \LL = 6 \,\genus(M) -
6$, the dimension of the Teichmuller space of $M$. 

Higher dimensional examples of negative Einstein spaces with nontrivial 
integrable deformation spaces are provided by K\"ahler-Einstein metrics. By a
result of Aubin and Yau, cf. \cite[Theorem 11.17]{besse:einstein}, 
any compact complex manifold with negative first
Chern class admits a unique K\"ahler-Einstein metric $(\bme,J)$,
where $J$ is the complex structure,  
with Einstein constant $\alpha <
0$. A result by Koiso \cite[Theorem 12.88]{besse:einstein} 
shows that the deformation space of $\bme$ is integrable if all infinitesimal
complex deformations of $J$ are integrable. Examples where these conditions
hold are provided  by hypersurfaces of ${\mathbb C} P^m$, $m \geq 3$, of
degree $d \geq m+2$, cf. \cite[Example 12.89]{besse:einstein}.

By a theorem of Dai et al. \cite{dai:etal:kahler}, compact K\"ahler-Einstein
spaces are stable. The proof uses the fact that K\"ahler-Einstein spaces
carry parallel spin${}^c$ spinors. In the presence of a parallel spinor $\sigma_0$, 
it is possible to relate $\LL$ to the square of the Dirac operator. Let
$\Sigma(M)$ denote the spin${}^c$ spinor bundle of $M$. Define
the map $\Phi : S^2 T^* M \to \Sigma(M) \tens T^* M $, taking symmetric
2-tensors to spinor-valued 1-forms, by 
$$
\Phi(h) =  h^i{}_j e_i \sigma_0 \tens e^j
$$
where $(e_i)$ is an ON frame on $(M,\bme)$ 
with dual frame $(e^i)$ and $h_{ij} =
h(e_i, e_j)$ are the frame components of $h$. Here indices are raised using
$\bme^{ij}$ which is just $\delta^{ij}$. A calculation shows 
$$
D^* D \,\Phi(h) = \Phi( \LL h - h \circ F+ \Ric\circ h)
$$
where $D$ is the Dirac operator, 
$F$ is the curvature of the line bundle appearing in the spin${}^c$ 
structure, and $(h \circ F)_{ab} = h_a{}^c F_{cb}$, and 
$(\Ric \circ h)_{ab}= \Ric_{ac} h^c{}_b$. 
The terms involving $F, \Ric$ give a positive
semi-definite contribution, and one can show 
$$
(\LL h , h )_{L^2} \geq (D\Phi(h), D\Phi(h) )_{L^2} \geq 0, 
$$
for all $h$, cf. \cite[Theorem 2.4]{dai:etal:kahler}. 
 
\section{Stability of product spaces} \label{sec:product} 
\subsection{Tensors on product spaces} 
Let $M$, $N$ be compact, connected, 
stable negative Einstein manifolds with constant $\alpha < 0$, 
of dimension $m,n$ and with metrics $\bme^M$, $\bme^N$ 
respectively. 
Then $M\times N$ is an Einstein space with constant $\alpha$
and metric $\bme = (\pi^M)^*\bme^M + (\pi^N)^* \bme^N$, where 
$\pi^M, \pi^N$ denote the projections of $M\times N$ to the
factors $M$, $N$, respectively. In the following we shall lift tensors on the 
factors $M,N$ 
to tensors on the product $M\times N$ by pulling back along the projections
$\pi^M$, $\pi^N$, and in order to avoid notational complications, for the
rest of this section we drop explicit reference to the projections. Thus, for
example, we
write the metric on of $M\times N$ simply as $\bme^M + \bme^N$. 

For tensors on $M$ we use greek indices $\mu,\nu,\gamma, \dots$, and for
tensors on
$N$, we use upper case latin indices $A,B,C,\dots$. 
For tensors on the
product $M\times N$, we use lower case latin indices $a,b,c,\dots$, or, when
convenient, a mixture of the two other index types. For
computations involving tensor products, we will use both index and index-free
notation. 
We will write $\nabla^M, \nabla^N$ for
the $M,N$ covariant derivatives, respectively, and $\nabla$ for the $M\times
N$ covariant derivative. Similarly, $\Delta^M$, $\Delta^N$ denote the
Laplacians on $M,N$ and $\Delta$ the Laplacian on $M \times N$. 
The space of symmetric covariant tensors of order
$k$ on $M$ will be denoted  by $\Sym^k M$, and the space of covariant $k$-tensors on $M$
will be denoted by $\Tens^k M$, and similarly for the other spaces. 
The covariant derivatives on $M$, $N$ apply in an
unambiguous way to lifted tensors. 
For example, let $\xi, \eta$ be tensors on $M,N$. Then 
with the above notation we have 
\begin{equation}\label{eq:ansatz}
\nabla(\xi\tens \eta) = (\nabla^M \xi )\tens \eta + \xi \tens ( \nabla^N
\eta) 
\end{equation} 
as well as the obvious index version of this formula. Similarly, we have 
$$
\Delta (\xi \tens \eta) = (\Delta^M \xi) \tens \eta + \xi \tens ( \Delta^N \eta) ,
$$
in particular there are no cross terms. 
Let $\odot$ be the
symmetric tensor product, by definition 
$$
\xi \odot \eta = \xi \tens \eta + \eta \tens \xi .
$$
Then 
$$
\nabla (\xi \odot \eta) = (\nabla^M \xi )\odot \eta + \xi \odot (\nabla^N
\eta)
$$
and analogously for $\Delta$. 

We will frequently consider symmetric tensors on $M\times N$ of the form 
$$
t = u \psi + \xi \odot \eta + \phi v
$$
where $u \in \Sym^2(M)$, $\xi\in \Tens^1(M)$, $\phi \in C^\infty (M)$, 
$\psi \in C^\infty(N)$, $\eta\in \Tens^1(N)$, $v \in
\Sym^2(N)$.  Then 
\begin{multline*}
\nabla t = \nabla^M u \psi + u \tens \nabla^N \psi 
+ (\nabla^M \xi )\odot \eta \\
+ \xi \odot (\nabla^N \eta)
+ \nabla^M \phi \tens v + \phi \nabla^N v
\end{multline*} 
and similarly for $\Delta$. For a Cartesian product $M\times N$, the cross
terms in the Riemann
tensor vanish in the sense that $R(X,Y,Z,W)=0$ whenever $X,Y,Z,W$ contains a
pair of vector fields which are tangent to $M$,  $N$ respectively. 
It follows that the operators $\cR$ defined in terms of the Riemann tensor
as in (\ref{eq:cRdef}) act 
block diagonally on
tensor products formed of tensors on $M$ and $N$. In particular with
$\xi,\eta$ as above, 
$$
\cR (\xi \odot \eta ) = 0
$$
Hence, we have 
\begin{align*} 
\LL t &= - u \Delta^N \psi  - \psi \Delta^M u - \Delta^M \xi \odot \eta - \xi
\odot \Delta^N \eta \\
&\quad - v \Delta^M \phi- \phi \Delta^N v 
- 2 \psi (\cR_M u)  - 2 \phi (\cR_N v) 
\end{align*} 
where $\cR_M, \cR_N$ are defined in terms the Riemann tensors  
$R_M, R_N$ of $M,N$, respectively. 
\subsection{Spectral decomposition} 
Consider the compact, negative Einstein spaces $M$, $N$ as above. 
%
Since we are considering compact manifolds, the 
covariant Laplacian
acting on functions and one-forms and $\LL$ acting on 
symmetric 2-tensors are self-adjoint, with
discrete spectrum. We will use the following notation for the spectral
decompositions. On $M$ the eigenvalues for the operators $-\Delta, -\Delta,
\LL$ acting on tensors of order $k=0,1,2$ 
will be denoted by $\lambda_i^{(k)}$. The $L^2$ normalized 
eigentensors of order $0,1,2$ will be denoted by $\phi_i, \xi_i, u_i$,
respectively. Similarly on $N$, the eigenvalues will be denoted
$\mu_i^{(k)}$, $k=0,1,2$, and the $L^2$ normalized 
eigentensors of order $0,1,2$ will be
denoted by $\psi_i, \eta_i, v_i$. Then from spectral theory, it follows that 
the $\{\phi_i\}, \{\xi_i\},
\{u_i\}$ and $\{\psi_i\}, \{\eta_i\}, \{v_i\}$ 
constitute $L^2$ bases for tensors of order $0,1,2$ on $M$ and $N$, 
respectively. 
\begin{lemma} \label{lem:complete} 
The tensor products $u_i \psi_j$, $\xi_i \odot \eta_j$,
$\phi_i v_j$ form a complete orthonormal system in $\Sym^2(M\times N)$. 
\end{lemma} 
\begin{proof} 
It is clear that $\{u_i \psi_j\}$, $\{\xi_i \odot \eta_j\}$ and $\{\phi_i
v_j\}$ are orthonormal sets. 
Let $f \in \Sym^2(M\times N)$. We can assume without loss of generality that
$f$ is smooth. Let $x^\mu$ and $y^A$ be coordinate systems on
$M,N$ so that $(x^\mu,y^A)$ is a coordinate system on $M\times N$. 
Then we can write $f = s_{\mu\nu}dx^\mu \odot dx^\nu + 
t_{\mu A} dx^\mu \odot dy^A + u_{AB} dy^A \odot dy^B$,  where the
coefficients are smooth functions of $x,y$. 
Now suppose $f$ is $L^2$ perpendicular
to all $u_i \psi_j$. This implies for all $i,j$,
$$
\int_N \left ( \int_M \la s(x,y) , u_i(x) \ra \mu_M(x) \right ) \psi_j
\mu_N(y) = 0 .
$$
Since $\{\psi_j\}$ is a basis for $L^2(N)$,
$$
\int_M \la s(x,y ), u_i(x) \ra \mu_M(x) = 0
$$
from which follows $s = 0$. We can deal with the other factors
similarly. 
\end{proof} 

\begin{lemma} \label{lem:prod-stable} Let $M,N$ be stable Einstein spaces
  with Einstein constant $\alpha < 0$. 
Then  $M\times N$ is stable and 
$$
\ker \LL^{M\times N} = \ker \LL^M + \ker \LL^N
$$
\end{lemma} 
\begin{proof} 
We will drop the superscript on
$\LL$ and other operators when it is clear from the context which space they
act on.
By Lemma \ref{lem:complete}, 
any tensor $t \in \Sym^2(M\times N)$ can be expanded in the form 
\begin{equation} \label{eq:expansion} 
t = \sum_{i,j} a_{ij} u_i \psi_j + \sum_{i,j} b_{ij} \xi_i \odot \eta_j +
\sum_{i,j} c_{ij} \phi_i  v_j 
\end{equation} 
with constants $(a_{ij}, b_{ij}, c_{ij})$.  A calculation shows 
\begin{align*} 
\LL t &=  \sum_{i,j} a_{ij} (\lambda_i^{(2)} + \mu_j^{(0)} ) u_i \psi_j \\
&+ \sum_{i,j} b_{ij} (\lambda_i^{(1)} + \mu_j^{(1)} ) \xi_i \odot
\eta_j \\
& + \sum_{ij} c_{ij} (\lambda_i^{(0)} + \mu_j^{(2)} ) \phi_i v_j 
\end{align*} 
Thus eigentensors of the operator $\LL$ on $M\times N$, can be of order
$(2,0)$, $(1,1)$, and $(0,2)$, with eigenvalues 
$\lambda_i^{(2)} + \mu_j^{(0)}$, $\lambda_i^{(1)} + \mu_j^{(1)}$, and
$\lambda_i^{(0)} + \mu_j^{(2)}$, respectively. 
Now assume each factor $M,N$ is stable, i.e. the operators $\LL^M, \LL^N$ are
nonnegative when acting on $\TT$ tensors. By Lemma \ref{lem:nonneg}, this
implies $\LL^M, \LL^N$ are
nonnegative when acting on all tensors on $M,N$ respectively. 
Further, the operator $-\Delta$
acting on functions is
nonnegative, and since $\Delta_H = -\Delta + \Ric$ is nonnegative, it follows
that $-\Delta$ acting on 1-forms has spectrum bounded from below by $-\alpha
> 0$. 
Hence, the operator $\LL$ on
$M\times N$ is nonnegative.  

It remains to identify the kernel of $\LL$. Suppose $\LL t = 0$. Recall that
$\LL$ acts on the off diagonal terms by $-\Delta$ which has spectrum bounded
from below by $-\alpha > 0$. Hence $t$ must have coefficients $b_{ij} = 0$. 
Examining the action of $\LL$ on $t$, we see that zero in the spectrum of
$\LL$ corresponds to zero in the spectrum of $\LL^M, \LL^N$ as well as zero
in the spectrum of $\Delta^M, \Delta^N$. Since the zero eigenfunction of the
scalar Laplacian is constant, we see that $t = u + v$, with $u \in \ker
\LL^M$ and $v \in \ker \LL^N$. This completes the proof of the theorem. 
\end{proof} 

The result of Lemma \ref{lem:prod-stable} clearly applies to
an arbitrary number of factors.

\subsection{Deformation spaces on products} 

In the following we will consider products with two factors. However, as in
the case of Lemma \ref{lem:prod-stable} the results in the rest of this
section apply to products with an arbitrary number of factors. 

We assume the deformation spaces $\NN^M$, $\NN^N$,  
of the background spaces 
$(M, \bme_0^M)$, $(N,\bme_0^N)$ are integrable and stable. 
We allow the case where one or both of
$\bme_0^M,\bme_0^N$ 
are strictly stable so that the corresponding deformation space is
trivial. 

\begin{prop} \label{prop:prod-integr} 
Let $(M,\bme_0^M)$, $(N,\bme_0^N)$ 
be stable Einstein spaces with Einstein
constant $\alpha$. Assume that $M,N$ have integrable deformation spaces. 
Then, the 
deformation space $\NN$ of $\bme$ is locally diffeomorphic to $\NN^M \times
\NN^N$. 
\end{prop} 
\begin{proof} 
Clearly $\NN$ contains the space $\NN^M \times \NN^N$. Let $\Slice$ be the
slice for $M\times N$ defined with respect to the product metric 
$\bme_0 = \bme_0^M + \bme_0^N$ and let $\Slice^M$, $\Slice^N$ be the
slices defined with respect to $\bme_0^M$, $\bme_0^N$. 
By Lemma \ref{lem:nonneg} and Lemma \ref{lem:prod-stable}, we have 
$T_{\bme_0} [ ( \VV^M \times \VV^N) \cap \Slice ] = \ker \LL = \ker
\LL^M + \ker \LL^N = T_{\bme_0^M}(\VV^M \cap \Slice^M) 
+ T_{\bme_0^N}(\VV^N \cap \Slice^N)$. 
Therefore, the  tangent space of the deformation space at $\bme_0$ is equal to the
formal tangent space.
It follows from \cite[Theorem 
13.49, p. 351]{besse:einstein}\footnote{Due to a typographical error in \cite{besse:einstein}, Theorem
  12.49 appears as Theorem 13.49.}
that the deformation
space is integrable near $\bme_0$. 
\end{proof} 
The following is an immediate corollary to Propositions \ref{prop:prod-integr} and \ref{prop:uniform}. 
\begin{cor} 
There is a neighborhood
  $\UU$ of $\bme_0$ in $\MM(M\times N)$ such that 
$\UU \cap \Ein_\alpha (M\times N) \cap \Slice = \NN^M \times \NN^N$. 
\end{cor}

\subsection{Examples of stable product spaces} \label{sec:examp-stab-prod}
The results above show that Cartesian products of spaces which 
are strictly stable, or stable and integrable, yield spaces which are stable
and integrable. If all factors are strictly stable, then the product is
strictly stable. Thus the examples discussed in section \ref{sec:examp-stab} 
allow us to construct large
families of stable integrable negative Einstein spaces with nontrivial
deformation spaces, as well as large families of strictly stable negative
Einstein spaces. 

Among the cases of interest are products of hyperbolic manifolds with
hyperbolic surfaces, as well as products of hyperbolic manifolds with negative 
K\"ahler-Einstein spaces.

\section{The Einstein evolution equations} \label{sec:rescaled}
Let $\bme \in \Ein_{-(n-1)/n^2}$. 
Then the Lorentz cone over $(M, \bme)$, i.e. the manifold 
$(0,\infty) \times M$ with metric
\begin{equation}\label{eq:lorcone} 
-dt\otimes dt + \frac{t^2}{n^2} \bme
\end{equation}
is a Ricci flat, maximal, globally hyperbolic spacetime, which admits a
timelike homothety $t \partial_t$.

In this section we write the Einstein evolution 
equations in terms of scale invariant
variables. As we shall see, the resulting system is 
autonomous, and  data corresponding to the Lorentz cone metric
(\ref{eq:lorcone}) is an equilibrium point for this system. 

\subsection{Scale invariant variables} \label{sec:scalinvvar}
Let $(\pme,\pK,\pLapse,\pShift)$ be constant mean curvature Cauchy data for the Einstein
equations. We use the same conventions as in \cite{AMF}. 
Let $\tau = \tr_{\pme} \pK$ denote the mean
curvature. We assume $\tau < 0$.  The rescaled variables corresponding to
$(\pme,\pK,\pLapse,\pShift)$
are 
$(\rme_{ij}, \rKtr_{ij}, \rLapse,
\rShift^i)$, defined by 
\begin{align*} 
\rme_{ij} &= \tau^2 \pme_{ij}, \qquad \rLapse = \tau^2 \pLapse \qquad  \rShift^i = \tau
\pShift^i \\ 
\rme^{ij} &= \frac{1}{\tau^2} \pme^{ij}, \qquad  
\rKtr_{ij} = \tau (\pK_{ij} - \frac{\tau}{n} \pme_{ij} ) 
\end{align*} 
Since we are assuming constant mean curvature, $\rKtr$ and $\tau$ 
contain the same information as $\pK$. Thus, it is natural to
use $(\rme,\rKtr,\rLapse,\rShift)$ as the set of scale invariant Cauchy
data. 
In particular, for the line element (\ref{eq:lorcone}), we have
$$
\pme(t) =  (t^2/n^2) \bme, \quad \pK(t) = -t^{-1} \pme, 
\quad \tau(t) = - n t^{-1}
$$ and hence the rescaled data are 
$$
(\rme,\rKtr)  \equiv (\bme,0).
$$ 

Note that  {\em by construction}  $\rKtr$ has vanishing trace. Thus, we view 
$(\rme,\rKtr)$ as an element of $T^{\tr} \MM$, where $T^{\tr} \MM$
denotes the subbundle of $T\MM$ with fiber at $g$ 
consisting of symmetric tensors $h$ such that $\tr_g h = 0$. 

%
Introduce the scale invariant
time $T$ by 
$$
T = - \ln(\tau/\tau_0) 
$$
where $\tau_0$ is some negative constant. Then, 
$\partial_T = -\tau \partial_\tau$.  
As we shall see below, the Einstein evolution 
equations in terms of the scale invariant
variables, and the scale invariant time $T$ form an autonomous system. 


\subsection{Constraint set and slice} \label{sec:slice} 
In this section, we review some results from \cite[\S
  2.3]{AMF} concerning the geometry of the constraint set and the 
slice for the diffeomorphism  group determined by the spatial harmonic gauge
condition. In \cite[\S 2.3]{AMF}, $M$ was assumed to have constant negative
curvature. However, the results which we shall need generalize immediately to
the case being considered here of negative Einstein spaces. 
In this paper we use the scale invariant metric and the shear
tensor $(\rme, \rKtr)$ as fundamental variables, so we shall use these in the
discussion of the constraint set and slice. 

Let $\CC$ be the space of
$(\rme, \rKtr) \in T^{\tr} \MM$
which are solutions to 
the vacuum constraint equations. Recalling that we are considering the
constant mean curvature case, the constraint equations written in terms 
of the rescaled variables
$(\rme, \rKtr)$ are of the form 
\begin{align*} 
0 &= \rR + \frac{n-1}{n} - |\rKtr|^2,  \\ 
0 &= \rnabla^i \rKtr_{ij} ,
\end{align*} 
where $|\rKtr|^2 = \rKtr_{ij} \rKtr^{ij}$. 
%
%
Fix a background metric $\bme_0 \in \Ein_\alpha$, let $\NN$ be the 
deformation space with respect to $\bme_0$, and 
let $\bme \in \NN$ 
be close to $\bme_0$.  
By the discussion in section \ref{sec:HC} 
there is a harmonic 
slice $\Slice_\bme \subset \MM$ defined with respect to $\bme$. 
Corresponding to $\Slice_\bme$, we have the slice $\SliceCC{}_{,\bme}$ 
in the constraint
set\footnote{in \cite{AMF}, $\SliceCC$ was denoted by $\Sigma$}, 
consisting of solutions $(\rme,\rKtr)$ to the constraint
equations such that $\rme \in \Slice_{\bme}$. 
By \cite[Lemma 2.3]{AMF}, which generalizes immediately to the present
situation,  
$\SliceCC{}_{,\bme}$ 
is a smooth submanifold of $T^{\tr} \MM$ with 
tangent space at
$(\bme,0)$ given by the affine subspace 
\begin{multline*}
T_{(\bme,0)} \SliceCC{}_{,\bme} = (\bme,0) \\
+ \{ (u^{\TT},v^{\TT}) ,  \text{ where }
u^{\TT}, v^{\TT} \ \text{ are $\TT$-tensors w.r.t. $\bme$} \}
\end{multline*} 
of $T_{\bme}^{\tr} \MM$. 
Given $\bme$, we have the decomposition 
\begin{equation}\label{eq:L2orth} 
t = t^{\TT} + [\phi \bme + \Lie_X \bme]
\end{equation} 
for any symmetric 2-tensor $t$, 
where $t^{\TT}$ is a $\TT$ tensor on $M$ with respect to $\bme$, and
$\phi$, $X$ are a function and a vector field on $M$, respectively. It is
important to note that the terms $t^{\TT}$ and $\phi \bme + \Lie_X\bme$ are
$L^2$-orthogonal. 
We can represent
$\SliceCC{}_{,\bme}$ as a graph over its tangent space. 
In particular, we may write $(\rme,\rKtr) \in \SliceCC{}_{,\bme}$ in the form 
\begin{equation}\label{eq:zw} 
\rme = \bme + u^{\TT} + z, \quad \rKtr = v^{\TT} + w
\end{equation} 
with $z,w$  $L^2$-orthogonal to the space of $\TT$ tensors defined with
respect to $\bme$. Then 
$(z,w)$ are second order in $(u^{\TT}, v^{\TT})$, i.e. an estimate of the
form 
$$
||z||_{H^s} + ||w||_{H^{s-1}} \leq C ( ||u^{\TT}||_{H^s}^2 +
||v^{\TT}||_{H^{s-1}}^2 )  
$$
holds.

\subsection{Gauge condition and the shadow metric} \label{sec:gauge-cond} 
We introduce the following modification of the CMCSH gauge of \cite{AML}. 
Let $\NN$ be the deformation space of $\bme_0$ and assume
that $\NN$ is integrable.

Introducing a system of local coordinates $(q^\alpha)$, $\alpha = 1, \dots, m
= \dim \NN$ on $\NN$, we may write
a general element of $\NN$ as $\bme_{ij} = \bme_{ij}(q^\alpha)$. Then 
$$
\frac{\partial \bme_{ij}}{\partial q^\alpha} 
= h_{ij}^{(\alpha)} 
$$
gives a basis for $T_{\bme} \NN$. 
From equation (\ref{eq:TNN}) we have that 
each 
$h_{ij}^{(\alpha)}$ admits a decomposition 
\begin{equation}\label{eq:halpha-first} 
h_{ij}^{(\alpha)} = 
h_{ij}^{(\alpha) \, \TTpara} + \Lie_{Y^{(\alpha)\, \parallel}} \bme_{ij} 
\end{equation} 
with $h_{ij}^{(\alpha) \, \TTpara} \in \ker \LL_\bme$, where $\LL_\bme$ is
defined by (\ref{eq:LLdef}) in terms of $\bme$. 
In general the  term 
$\Lie_{Y^{(\alpha) \, \parallel}} \bme_{ij}$,
where $Y^{(\alpha) \, \parallel}$ is determined in terms of 
$h^{(\alpha)\, \TTpara}$ by equation (\ref{eq:PYi}), 
is non-vanishing.

For $\rme \in \MM$, $\bme \in \NN$, we say that $\rme - \bme$ is
{\bf $L^2$-perpendicular} to $\NN$ at $\bme$ if the conditions
$$
0 = (\rme - \bme, h^{(\alpha)} )_{L^2;\bme} ,\quad \alpha = 1,\dots,m
$$
or explicitly, 
\begin{equation}\label{eq:shadowdef-first} 
0 = \int_M (\rme_{ij} - \bme_{ij} ) \frac{\partial \bme^{ij}}{\partial
q^\alpha} \mu_{\bme} , \quad \alpha = 1, \dots, m
\end{equation} 
hold. 
Here 
$$
\frac{\partial \bme^{ij}}{\partial q^\alpha} = - \left ( 
\frac{\partial \bme_{kl}}{\partial q^\alpha} \right ) \bme^{ki} \bme^{lj} .
$$

\begin{remark} \label{rem:shade} Equation (\ref{eq:shadowdef-first}) can be
  viewed as defining a smooth projection 
map $\Shade$ mapping a
  neighborhood of $\bme_0$ in $\MM$ to $\NN$, given by $\Shade[\rme] =
  \bme$. See section \ref{sec:shadow}  for the relevant calculation. 
We refer to the map $\Shade$ as the {\bf shadow map}. In view of the fact,
cf. \cite[Chapter 5F]{besse:einstein} that
Einstein metrics are smooth (in fact real analytic) in harmonic coordinates,
it follows that $\Shade$ is {\bf smoothing} in the sense that, for $\rme$ 
in the neighborhood
of $\bme_0$ where $\Shade$ is defined and regular, 
we have that $D\Shade \big{|}_{\rme}: H^s \to H^{s'}$ is continuous for any $s, s'$. 
\end{remark} 


\begin{definition}\label{def:shadow-gauge-new}
Fix $\bme_0 \in \Ein_{-(n-1)/n^2}$, and let $\NN$ be the 
deformation space defined with respect to $\bme_0$, see section
\ref{sec:slicedeform}.  Assume that $\NN$ is integrable.
We say that for $\bme \in \NN$, and Cauchy data $(\rme,\rKtr) \in \CC$, 
the triple  
$(\bme,\rme,\rKtr)$ satisfies the {\bf shadow metric condition} 
if 
\begin{enumerate}
\item 
$\rme$ 
satisfies the {\bf CMCSH gauge condition} with respect to $\bme$, i.e. 
\begin{equation}\label{eq:rgauge} 
V^i = 0
\end{equation} 
where $V$ is the tension field defined with respect to
$\rme$ and $\bme$, cf. equation (\ref{eq:Vdef}), 
\item
$\rme - \bme$ is $L^2$-perpendicular to the tangent space $T_{\bme} \NN$ of
$\NN$ at $\bme$, in the sense of (\ref{eq:shadowdef-first}). 
\end{enumerate} 
If this holds, we call $\bme$ the {\bf shadow metric} of
$\rme$. 
\end{definition} 
The following Lemma, proved in appendix \ref{sec:shadow:gauge}, 
shows that the shadow metric condition can always be
satisfied locally. 
\begin{lemma}\label{lem:shadow-exist-met}
  Let $\NN$ be
the deformation space defined with respect to $\bme_0$. Assume $\NN$ is
integrable. Let $s > n/2+1$.  There is a $\delta > 0$ such that for 
$\rme \in \MM$ satisfying $||\rme - \bme_0||_{H^s} < \delta$,
there exists  a unique 
$\bme$ in $\NN$ satisfying $|| \rme - \bme||_{H^s} < 2 \delta$,
such that $\bme$ is the shadow metric of $\rme$. 
\end{lemma} 

Let $\Shade$ be the shadow map introduced in remark \ref{rem:shade}. For
$\rme \in \MM$ close to $\bme_0$, 
such that $\Shade$ is well defined near $\rme$, 
define the operator $Q_{\rme,\bme}$ acting on tensors $z$ by 
\begin{equation}\label{eq:Qdef}
(Q_{\rme,\bme} z)^i = \rme^{mn} D \Gamma_{mn}^i \big{|}_\bme . D \Shade \big{|}_\rme
. z
\end{equation} 
The following Lemma will be needed in the proof of local existence, see section
\ref{sec:locex}, as well as in the proof of Lemma \ref{lem:quad-NX-est}. 
\begin{lemma} \label{lem:Qlemma} Let $s > n/2 + 1$. 
Let $\bme_0 \in \Ein_\alpha$ be given, assume $\bme_0$ has 
stable, integrable deformation space $\NN$, and for $\rme \in \MM$
with shadow metric $\bme \in \NN$, let 
$Q_{\rme,\bme}$ be defined by (\ref{eq:Qdef}). 
There is a $\delta > 0$ such that 
if $||\rme - \bme_0||_{H^s} < \delta$,  
the inequality
$$
||Q_{\rme,\bme} z ||_{H^s} \leq C(||\rme - \bme||_{H^s} + ||\bme -
\bme_0||_{H^s} ) ||z||_{H^{s-1}}
$$
holds. 
\end{lemma} 
\begin{proof} 
We have 
\begin{align*} 
(Q_{\rme,\bme} z)^i &= 
(\rme^{mn} - \bme^{mn}) D \Gamma_{mn}^i \big{|}_\bme . D \Shade \big{|}_\rme
. z \\
&\quad 
+ \bme^{mn} D \Gamma_{mn}^i \big{|}_\bme . D\Shade \big{|}_\rme . z
\end{align*} 
By the discussion in section \ref{sec:slicedeform}, $D\Shade \big{|}_\rme
. z$ is of the form 
$$
h^{\TT\parallel} + \Lie_{Y^\parallel} \bme
$$
A calculation shows that 
\begin{equation}\label{eq:DGammauTT} 
\bme^{mn} D \Gamma_{mn}^i \big{|}_\bme . u^{\TT} = 0.
\end{equation} 
for any tensor $u^{\TT}$ which is $\TT$ with respect to $\bme$. In particular 
$\bme^{mn} D \Gamma_{mn}^i \big{|}_\bme . h^{\TTpara} = 0$. 
Further, $Y^\parallel$ solves (\ref{eq:PYi}). As discussed in section
\ref{sec:HC}, the operator $P$ occurring in that equation is an
isomorphism, in fact as an operator $H^{s'} \to H^{s'-2}$ for any $s'$. 

This allows us to estimate 
$$
|| 
\bme^{mn} D \Gamma_{mn}^i \big{|}_\bme . D\Shade \big{|}_\rme . z 
||_{H^s}
 \leq C ||\bme - \bme_0||_{H^s} ||z||_{H^{s-1}}
$$
Here we have made use of the fact that the shadow map 
$\Shade$, in view of its definition,
is smoothing, see remark \ref{rem:shade}. 
Together with the above discussion this completes the proof.
\end{proof}

\subsection{The scale invariant evolution equations} \label{sec:scaleinv:evol}
Define 
$$
\del_{ij} = \half ( \rnabla_i V_j + \rnabla_j V_i) ,
$$
where $V$ is the tension field defined with respect to $\rme$, $\bme$. 
Following the work in \cite{AML} we will
consider the modified Einstein evolution equations 
obtained by replacing $\rR_{ij}$ by the
quasilinear elliptic system 
$$
\rR_{ij} - \del_{ij} .
$$
We remark that both $\rR_{ij}$ and $\del_{ij}$ are scale invariant quantities. 
The modified Einstein evolution
equations that we will consider are, 
in terms of the scale invariant variables, 
\begin{subequations}\label{eq:Einrescaled} 
\begin{align} 
\rme_{ij,T} &= 2\rLapse \rKtr_{ij} + 2(\frac{\rLapse}{n} - 1) \rme_{ij} -
\Lie_{\rShift} \rme_{ij} \label{eq:evolution-tg} \\
\rKtr_{ij,T} &= -(n-1) \rKtr_{ij} 
 - \rLapse ( \rR_{ij} - \del_{ij} + \frac{n-1}{n^2} \rme_{ij} ) \nonumber \\
&\quad + \rnabla_i \rnabla_j \rLapse + 2 \rLapse \rKtr_{im} \rKtr^m_{\ j}
 \nonumber \\
&\quad  - \frac{1}{n} (\frac{\rLapse}{n} - 1) \rme_{ij} 
- (n-2) (\frac{\rLapse}{n} - 1) \rKtr_{ij} \nonumber \\ 
&\quad 
- \Lie_{\rShift} \rKtr_{ij}  \label{eq:evolution-tKtr} 
\end{align} 
\end{subequations}
and the constraint equations take the form
\begin{subequations}\label{eq:Conrescaled}
\begin{align} 
0 &= \rR + \frac{n-1}{n} - |\rKtr|^2 \label{eq:Conrescaled-ham} \\ 
0 &= \rnabla^i \rKtr_{ij} \label{eq:Conrescaled-mom}  
\end{align} 
\end{subequations} 
The defining equations 
for the rescaled lapse and shift are
\begin{subequations}\label{eq:defining-resc} 
\begin{align} 
-\rDelta \rLapse + (|\rKtr|^2 + \frac{1}{n} ) \rLapse &= 1 \label{eq:Lapresc} \\
\rDelta \rShift^i + \rR^i_{\ j} \rShift^j - \Lie_{\rShift} V^i &= 
(-2\rLapse \rKtr^{mn} + 2 \rnabla^m \rShift^n) (\Gamma[\rme]_{mn}^i 
- \bGamma_{mn}^i) \nonumber \\
&\quad 
- 2( \frac{\rLapse}{n}-1) V^i \nonumber \\
&\quad + 2 (\rnabla^m \rLapse) \rKtr^i_{\ m} 
+ (2 - n) \rnabla^i
( \frac{\rLapse}{n}-1) 
\nonumber \\ 
&\quad - \rme^{mn} \partial_T \bGamma_{mn}^i \label{eq:shiftresc} 
\end{align} 
\end{subequations} 
where the last term will be present only in case $\bme_0$ has a nontrivial
deformation space $\NN$. 

We remark that the evolution equations (\ref{eq:Einrescaled}) 
do not automatically leave invariant the
conditions $\tr \rKtr =0$ and $V = 0$,  
but these hold as a consequence of imposing the
defining equations (\ref{eq:defining-resc}), cf. the discussion in
section \ref{sec:locex}. 

\subsection{Evolution of the shadow metric} 
\label{sec:shadow} 
Recall from section \ref{sec:gauge-cond} that the 
shadow metric is defined by the relation 
\begin{equation}\label{eq:shadowdef} 
0 = \int_M (\rme_{ij} - \bme_{ij} ) \frac{\partial \bme^{ij}}{\partial
q^\alpha} \mu_{\bme} , \quad \alpha = 1, \dots, m
\end{equation} 
We now time differentiate (\ref{eq:shadowdef}). 
Note that we can write 
$$
\partial_T \bme_{ij} = \frac{\partial \bme_{ij}}{\partial q^\beta} \dot
q^\beta
$$
Thus time differentiating (\ref{eq:shadowdef}) gives 
\begin{align*} 
0 &= \int_M (\rme_{ij} - \bme_{ij}) 
\frac{\partial}{\partial q^{\beta}} \left ( \frac{\partial \bme^{ij}}{\partial
q^\alpha} \right ) {\dot q}^\beta \mu_{\bme} \\ 
&\quad + \int_M (\rme_{ij} - \bme_{ij}) \frac{\partial \bme^{ij}}{\partial
q^\alpha} \half \bme^{mn} \frac{\partial \bme_{mn}}{\partial
q^\beta} {\dot q}^\beta \mu_{\bme} \\
&\quad + \int_M \partial_T \rme_{ij} \frac{\partial \bme^{ij}}{\partial
q^\alpha} \mu_{\bme} \\
&\quad + \int_M {\dot q}^\beta \left ( \frac{\partial
\bme_{ij}}{\partial q^\beta} \bme^{im} \bme^{jn} \frac{\partial
\bme_{mn}}{\partial q^\alpha} \right ) \mu_{\bme} 
\end{align*} 

The matrix 
$$
M_{\alpha\beta} := \int_M \frac{\partial \bme_{ij}}{\partial
q^\beta} \bme^{im} \bme^{jn} \frac{\bme_{mn}}{\partial q^\alpha}
\, \mu_{\bme}
$$ 
is invertible since the $\frac{\partial \bme_{ij}}{\partial q^\alpha}$ form
a basis for the tangent space $T_{\bme} \NN$. 

Thus for $(\rme_{ij} - \bme_{ij})$ sufficiently small, 
\begin{align*} 
\tM_{\alpha\beta} &:= \int_M \left ( \frac{\partial \bme_{ij}}{\partial
q^\beta} \bme^{im} \bme^{jn} \frac{\partial \bme_{mn}}{\partial q^\alpha}
\right ) \, \mu_{\bme} \\
&\quad + \int_M (\rme_{ij} - \bme_{ij} ) \frac{\partial}{\partial q^\beta}
\left ( \frac{\partial
\bme^{ij}}{\partial q^\alpha} \right ) \, \mu_{\bme} \\
&\quad + \int_M (\rme_{ij} - \bme_{ij} ) \frac{\partial \bme^{ij}}{\partial
q^\alpha} \half  \bme^{mn} \frac{\partial \bme_{mn}}{\partial
q^\beta} \, \mu_{\bme} 
\end{align*} 
will also be invertible and positive definite 
so we can solve for $\dot q^\beta$ in terms of the
inverse matrix applied to $\int_M \partial_T \rme_{ij} \frac{\partial
\bme^{ij}}{\partial q^\alpha} d \mu_{\bme}$. The defining equation for
$\partial_T \bme_{ij}$ can now be expressed in terms of an equation for
${\dot q}^\beta$ of the form 
$$
0 =  \tM_{\alpha\beta} \dot q^\beta + \int_M \partial_T \rme_{ij} 
\frac{\partial \bme^{ij}}{\partial q^\alpha} \, \mu_{\bme}
$$
In applying this setup, 
$\partial_T \rme_{ij}$ will be given by (\ref{eq:evolution-tg}). 

\begin{remark} \label{rem:shade-der} 
In terms of the shadow map $\Shade: \MM \to \NN$, cf. remark \ref{rem:shade}, we have 
$$
\partial_T \bme = D \Shade \big{|}_{\rme} \partial_T \rme .
$$
\end{remark} 

\section{Local existence} \label{sec:locex}
We shall prove local wellposedness for the system
(\ref{eq:Einrescaled}, \ref{eq:Conrescaled}, \ref{eq:defining-resc}) for the
rescaled variables $(\rme,\rKtr,\rLapse,\rShift)$ by making use of the results
of \cite{AML} applied to the corresponding system of modified Einstein
evolution equations, constraint equations and defining equations 
for the original variables
$(\pme, \pK,\pLapse,\pShift)$. Note that in \cite{AML} these fields are
denoted without the tilde. 

If the deformation space is trivial, i.e. if $\bme_0$ is strictly stable, 
we can apply the same argument as in the standard case considered in
\cite{AML}. 
We will now consider the modifications necessary for the case when the
deformation space is nontrivial. Let $t$ be the time variable in the system
considered in \cite{AML}. For the solution to this system, it will be the
case that $t$ is the CMC time, $t = \tr K$. 
However, it should be noted that 
in the course of the iteration procedure, and the proof of well posedness, as
presented in \cite{AML}, this cannot be assumed. 
We remark that it follows from
our assumption on initial data, and continuity, that 
$\pme^{ij} \pK_{ij}/t$ is close to $1$ for times $t$ close to the initial
time, 
so that in the construction of solutions we may assume that we 
are in an almost CMC situation. 

For the application in this paper, we are
interested only in the small data situation where 
$t^2 \pme$ is close to a background metric $\bme_0$, and 
$t\pK_{ij} - (t^2/n) \pme_{ij}$ 
is small. 

Define $\taubar$ to be the average mean curvature, 
$$
\taubar = (\int_M \pme^{ij} \pK_{ij} \mu_{\pme} ) / \int_M \mu_\pme
$$

The system of equations which will be
considered is the same as in 
\cite{AML}, with the difference that the 
shadow metric $\bme$ is time dependent, $\bme = \bme(t)$, and that the
spatially harmonic gauge is now defined with respect to the time dependent
shadow metric 
$\bme$. 
We define the shadow metric using the shadow map $\Shade$, see 
remark \ref{rem:shade}, by letting 
the shadow metric $\bme(t)$ be given by 
$$
\bme(t) = \Shade[\taubar^2(t) \pme(t)] .
$$
As mentioned above, we are considering a small data situation and in
particular, $\taubar^2 \pme$ is in a neighborhood of the shadow manifold $\NN$
where the map $\Shade$ is well defined 
and smooth. 

Using the notation of \cite{AML}, 
the system of evolution equations can be written in the form 
$$
L[\pme,\pLapse,\pShift] \UU = \FF
$$
with $\UU = (u,v)$. 
Let $\bnabla$ be the covariant derivative defined with respect to the metric
$\bme$. 
In order to prove local wellposedness for the resulting system, we use a
wave-equation type energy analogous to the one used in
\cite{AML}, i.e. an energy of the form 
$$
E = \int_M (|u|^2 + |\bnabla u|_\pme^2 + |v|^2)\mu_{\pme}
$$
where for a 2-tensor $u$, $|u|$ is defined in terms of $\bme$ by 
$|u|^2 = u_{ij} u_{kl} \bme^{ik} \bme^{jl}$, and $|\bnabla u|_\pme$ is defined
by $|\bnabla u|_\pme^2 = \bnabla_i u_{jk} \bnabla_l u_{mn} \pme^{il} \bme^{jm}
\bme^{kn}$.

Given $\pme,\pShift$, let $\rho$ be defined by
\begin{equation}\label{eq:rhodef}
\rho = -\half (\partial_t \pme - \Lie_\pShift \pme ) ,
\end{equation}
Taking into account the time dependence of $\bme$, the energy estimate
\cite[Lemma 2.4]{AML} is replaced by 
$$
\partial_t E \leq C (E^{1/2} ||\FF||_{H^1\times L^2} + (1 + ||\rho||_{L^\infty}
+ ||\partial_t \bme||_{L^\infty} ) E
$$
Using this estimate it is straightforward to prove the higher order energy
estimates needed for local wellposedness, following the same argument as in
\cite{AML}. 
Once the energy estimate is obtained, a solution to the field equations is
constructed using an iteration as in \cite[\S
2.2]{AML}. 

The argument from \cite{AML} carries over nearly without
modification. 
However, the fact that 
we now allow for a moving
shadow metric $\bme$ 
gives rise to an extra term in the defining equation for
the shift vector, of the form 
$$
\pme^{mn} \partial_t \Gamma[\bme]^i_{mn}.
$$
Since $\partial_t \bme = D\Shade\big{|}_{\taubar^2 \pme} \partial_t
(\taubar^2 \pme)$, and $\partial_t \pme =
- 2\pLapse \pK + \Lie_\pShift \pme$, 
this adds a non-local operator acting on $\pShift$ to the
defining equation for the shift vector. 

However, in view of the estimate given in
Lemma \ref{lem:Qlemma}, the effect of this term is a small
perturbation for the small data situation we are considering here, and hence
the modified shift equation satisfies the same estimates as the one which was
considered in \cite{AML}. 
In order to analyze the modified shift 
operator, 
it suffices to
consider the expression 
$$
\pme^{mn} D \Gamma^i_{mn} \big{|}_\bme . D\Shade \big{|}_{\taubar^2 \pme}
(\Lie_{\taubar \pShift}
 \taubar^2 \pme) .
$$
Here and below 
we have made use of the average mean curvature to introduce in appropriate places scale invariant
fields along the lines of section 
\ref{sec:scalinvvar}.
We shall apply Lemma \ref{lem:Qlemma} to estimate 
$$
Q_{(\taubar^2 \pme, \bme)}^i . (\Lie_{\taubar \pShift} \taubar^2 \pme) = 
\taubar^{-2} \pme^{mn} D \Gamma^i_{mn} \big{|}_\bme . D\Shade \big{|}_{\taubar^2 \pme} 
(\Lie_{\taubar\pShift} \taubar^2 \pme) .
$$
We have, in the small data situation we are considering,  
\begin{align} 
||Q_{(\taubar^2 \pme, \bme)} .  \Lie_{\taubar \pShift} \taubar^2 \pme ||_{H^s} 
&\leq C ( ||\taubar^2 \pme - \bme||_{H^s} +
||\bme - \bme_0||_{H^s} ) ||\Lie_{\taubar \pShift} \taubar^2 \pme||_{H^{s-1}} \\
&\leq C( ( ||\taubar^2 \pme - \bme||_{H^s} +
||\bme - \bme_0||_{H^s} ) ||\taubar \pShift||_{H^{s+1}} \, , 
\label{eq:Qappl}
\end{align}
where in the final line we stated the inequality in terms of 
$||\taubar \pShift||_{H^{s+1}}$ since that is
the norm which is relevant for the application we have in mind. 
Note we are considering only the
case where $t < t_0 < 0$ for some $t_0$, and consequently due to the small
data assumption, $\taubar < \taubar_* < 0$ for some $\taubar_*$.
In view of this estimate, the extra
term in the shift equation due to the time dependence of $\bme$ can be
considered as a small perturbation which does not affect the existence and 
uniqueness results for this equation proved in \cite{AML}. In particular, the
estimates for $\pShift$ needed for the iteration argument in \cite{AML} are valid
also for the modified system. 

With this remark, the rest of the argument goes
through unchanged. 
This proves local wellposedness in $H^s$, $s > n/2+1$. 

It remains to consider the propagation of gauges and constraints. 
Introduce, as in \cite{AML}, the gauge and constraint quantities 
\begin{subequations} \label{eq:defConstr}
\begin{align}
\tilde A &= \ptr \pK - t , \\
\tilde V^k &= \pme^{ij} ( \Gamma[\pme]_{ij}^k - \Gamma[\bme]_{ij}^k ) \, ,\label{eq:Vdef-late}\\
\tilde F &= R[\pme] + (\ptr \pK)^2 - |\pK|^2_\pme
- \nabla[\pme]_i \tilde V^i ,\\
\tilde D_i &= \nabla[\pme]_i \ptr \pK - 2\nabla[\pme]^j \pK_{ji} ,
\end{align}
\end{subequations}
Here we raise and lower indices using $\pme^{ij}$ and $\pme_{ij}$ and use the
notation $\ptr \pK = \pme^{mn} \pK_{mn}$ and $|\pK|_\pme^2 = \pK_{mn}
\pK^{mn}$. 
We consider the energy expression 
$\widetilde{\CMcal E}$ defined by 
\begin{align*} 
\widetilde{\CMcal E} &= \half \int_M (|\tilde A|^2 + |\nabla[\pme] \tilde
A|_\pme^2 + |\tilde F|^2 )\mu_\pme 
+ \half \int_M (|\tilde V|_\pme^2 + |\nabla[\pme]\tilde V|_\pme^2 + |\tilde D|_\pme^2 )\mu_\pme
\end{align*} 
Recall that the defining equations for $\pLapse$ and $\pShift$ are derived by
time differentiating the gauge conditions using the evolution equations. In
particular, the shift equation may be chosen such that the expression for
$\partial_t \tilde V$ is the same as in the rigid case. 
One sees
from considering the gauge and constraint quantities that the only 
potential 
difference
from the case when the deformation space is trivial 
is via the evolution equation for $\tilde V$. 
Thus, with the aforementioned choice, the
evolution equations for the gauge and constraint quantities are identical
in 
the case when the deformation space is nontrivial, to the evolution equations 
valid in the case when the
background metric $\bme_0$ is rigid. 
Therefore we are able to conclude
by the  same argument as in \cite{AML}, that if the constraint and gauge
conditions are
satisfied initially, they are also satisfied throughout the course of the
evolution. 

The system discussed here is the modified Einstein evolution equations
without rescaling. 
We now return to the situation considered in the rest of the paper
and state the results we have proved for the 
rescaled variables $(\rme,\rKtr,\rLapse,\rShift)$ 
introduced in section
\ref{sec:rescaled}.  It follows from the above discussion that 
the system of  
equations for the rescaled variables
(\ref{eq:Einrescaled}, \ref{eq:Conrescaled}, \ref{eq:defining-resc}) 
is also well
posed
in the shadow metric gauge.
As shown in \cite{AML}, the result that we have proved
can be formulated as a continuation principle, which will be used for the
global existence theorem. 

In formulating the continuation principle we will deal with rescaled data. 
Fix a background metric 
$\bme_0 \in \NN$, and for $\delta > 0$, $s > n/2 + 1$, let 
$\BB_{s, \delta}(\bme_0, 0)$ 
be the ball of radius $\delta$ 
in $\CC^s$, centered on $(\bme_0, 0)$. Here $\CC^s$ denotes the space of
$(\rme, \rKtr) \in \CC$ such that $(\rme, \rKtr) \in H^s \times H^{s-1}$. 
The following Lemma follows immediately from Lemma 
\ref{lem:shadow-exist-met} 
and
the construction of the shadow map $\Shade$.
\begin{lemma} \label{lem:DShade} 
Assume $\NN$ is integrable near $\bme_0$. 
Then there is a $\delta > 0$ such 
that 
\begin{enumerate} 
\item for $(\rme, \rKtr) \in \BB_{s, \delta}(\bme_0, 0)$, 
there is a unique $\bme \in
\NN$ such that the triple $(\bme, \rme, \rKtr)$ satisfies the shadow metric
condition, 
\item there is a constant $C$ such that  
the Frechet derivative $D\Shade$ satisfies 
$$
||D\Shade||_{Op,s,\infty} \leq C
$$
in $\BB_{s, \delta}(\bme_0, 0)$. 
Here the left hand side is the operator norm of
the Frechet derivative from $H^s$ to $L^\infty$.
\end{enumerate} 
\end{lemma} 
Based on  the above discussion of the proof of local wellposedness and the
continuation principle \cite[Theorem 5.1]{AML}, we can state the
following continuation principle which is appropriate for the situation
considered in this paper. Note that in the small data situation considered
here, the assumption of negative sectional curvature for the target metric
made in \cite{AML} is not needed. In particular, if $\bme$ is negative
Einstein and $\rme$ is close to $\bme$, then the operator $X^i \mapsto 
\rme^{mn} R[\bme]^i{}_{mjn} X^j$ is strictly negative, in which case the
proof of \cite[Lemma 5.2]{AML} goes through unchanged.  
\begin{thm} \label{thm:continuation} Fix $s > n/2 + 1$. 
Let $\bme_0 \in \NN$ be given and assume $\NN$ is integrable 
near $\bme_0$. Then there is a $\delta_0 > 0$ such that the conclusion of 
Lemma \ref{lem:DShade} holds, and such that 
the following continuation principle holds for the system
(\ref{eq:Einrescaled}, \ref{eq:Conrescaled}, \ref{eq:defining-resc}):

Let 
$(\rme_0, \rKtr_0)$ be rescaled data given at an initial rescaled time $T_0$,
such that the triple 
$(\bme_0, \rme_0, \rKtr_0)$ satisfies the shadow metric condition. Assume
that 
\begin{enumerate} 
\item 
$(\rme_0, \rKtr_0) \in
\BB_{s,\delta_0}(\bme_0,0)$, 
\item 
$[T_0, T_+)$ is the maximal existence
interval in $H^s$ for the system 
(\ref{eq:Einrescaled}, \ref{eq:Conrescaled}, \ref{eq:defining-resc}) with
initial data $(\bme_0, \rme_0, \rKtr_0)$. 
\end{enumerate} 
Then, either 
\begin{enumerate} 
\item $T_+ = \infty$ or 
\item
the solution curve $T \mapsto (\rme_{ij}(T),
\rKtr_{ij}(T))$ leaves $\BB_{s,\delta_0}(\bme_0,0)$ at some finite time.
\end{enumerate} 
\end{thm}

\section{Definition of energies}
\label{sec:linein-energ}  
In this section, assume that a background $\bme_0$ is given,
  with integrable deformation space $\NN$, and let $\delta_0 > 0$ be a
  sufficiently small constant such that for $(\rme,\rKtr) \in
B_{s,\delta_0}(\bme_0,0)$, the
  conclusion of Theorem \ref{thm:continuation} holds. For the estimates
  proved below in this section, we suppose that $\delta_0$ is decreased as
  necessary.

\subsection{Small quantities} \label{sec:small} 
The data corresponding to the Lorentz cone metric 
  (\ref{eq:lorcone}) is 
$$
(\rme,\rKtr,\rLapse,\rShift) = (\bme,0,n,0)
$$
Thus, 
$$
\rme-\bme, \rKtr, \frac{\rLapse}{n} -1, \rShift
$$
should be considered as small quantities in the small data situation we are
considering. 

In this section we shall consider a solution to the rescaled Einstein
equations presented above, satisfying the CMCSH gauge conditions
(\ref{eq:rgauge}) with respect to a curve of background metrics $\bme$
defined by the shadow metric condition. In case $\bme_0$ is strictly stable,
$\bme \equiv \bme_0$ and hence $\partial_T \bme = 0$. We will show that in this
case the
quantities $\rLapse/n-1, \rShift$ satisfy quadratic estimates in terms of 
$\rme-\bme,\rKtr$. 

If $\bme_0$ has a nontrivial integrable
deformation space, then 
\begin{equation}\label{eq:dTbme-form} 
\partial_T \bme = h^{\TTpara} + \Lie_{Y^\parallel} \bme
\end{equation}
is in general non-vanishing,
and in particular $Y^{\parallel} \ne 0$. In this case,  
$\rLapse/n -1$
satisfies a quadratic estimate as in the strictly stable case, while
$\rShift, Y^\parallel$ separately cannot be expected to satisfy such an 
estimate. However, as we shall prove, 
the sum $\rShift + Y^\parallel$ does. This is
precisely what is needed for the energy estimates to go through. 

\begin{lemma}\label{lem:quad-NX-est} Let $s > n/2+1$. 
For $(\rme,\rKtr) \in
B_{s,\delta_0}(\bme_0,0)$,
there is a constant $C > 0$ such that the  inequalities 
\begin{equation}\label{eq:hLapse-est}
||\frac{\rLapse}{n}-1||_{H^{s+1}} \leq C ||\rKtr||_{H^{s-1}}^2 ,
\end{equation}
\begin{equation}\label{eq:rS-est}
||\rShift||_{H^{s+1}} \leq C ( ||\rKtr||_{H^{s-1}} + ||\rme - \bme||_{H^s} ) 
\end{equation}
and 
\begin{equation}\label{eq:rSY-est}
||\rShift + Y^{\parallel}||_{H^{s+1}} \leq C ( ||\rKtr||_{H^{s-1}}^2  + 
||\rme-\bme||_{H^s}^2 ) 
\end{equation} 
hold. 
\end{lemma} 
\begin{proof} 
Recall that the scale invariant lapse and shift $\rLapse,\rShift^i$ solve the
  defining equations (\ref{eq:defining-resc}). 
It follows from the lapse equation (\ref{eq:Lapresc}) that 
$\hLapse= \rLapse/n - 1$ satisfies the equation 
$$
- \rDelta \hLapse + (|\rKtr|^2 + \frac{1}{n} ) \hLapse =  - |\rKtr|^2 
$$
From this, one finds
that if the gauge conditions are satisfied,
$\hLapse$ can be estimated in
terms of square norm of 
$\rKtr$, which proves (\ref{eq:hLapse-est}).

In case $\bme_0$ is rigid, then $\partial_T \bme = 0$ and 
$\rShift$ can be
estimated in terms of square norms of $\rme - \bme$ and $\rKtr$, 
cf. \cite[\S 3]{AML}. However, in case
$\bme_0$ has a nontrivial deformation space $\VV$, $\partial_T \bme$
will in general be nonzero. 
By (\ref{eq:evolution-tg}) and (\ref{eq:hLapse-est}), we have an estimate of the form 
\begin{equation}\label{eq:dtrme-est}
||\partial_T \rme||_{H^{s-1}} \leq C ( ||\rKtr||_{H^{s-1}}  + ||\rShift||_{H^{s}} )
\end{equation} 
As discussed in section \ref{sec:stabein}, we may without loss of generality
assume that $\bme_0$, and all metrics in $\NN$ are smooth. 
Hence, it follows from the
definition of the shadow metric that $\partial_T \bme$ is smooth and that
an estimate of the form 
$$
||\partial_T \bme ||_{H^{s'}} \leq C ||\partial_T \rme ||_{L^2}
$$
holds for any $s' \geq 0$. Consider the last term in equation
(\ref{eq:shiftresc}). Following the proof of Lemma \ref{lem:Qlemma} it
is straightforward in view of the above discussion to show that an 
inequality of the form 
\begin{multline}\label{eq:lastterm-est}
||\rme^{mn} \partial_T \bGamma_{mn}^i ||_{H^s} \leq \\
C (||\rme - \bme||_{H^s}
+ ||\bme - \bme_0||_{H^s} ) (||\rKtr||_{H^{s-1}} + ||\rShift||_{H^{s+1}})
\end{multline} 
holds. Here, as in (\ref{eq:Qappl}),  we stated the inequality in terms of $||\rShift||_{H^{s+1}}$
since that is the norm which is relevant for the application at hand. 

We now recall some facts from \cite{AML}. Let $V_{\rme,\bme}$ be the tension
field defined with respect to the metrics $\rme,\bme$, and denote by 
$\rnabla[g],\rDelta[g],\rR[g]$ 
the covariant derivative, Laplacian and curvature defined
with respect to $\rme$, and by $\bnabla, \bDelta,\bR$ the corresponding objects
defined with respect to $\bme$.  

Let 
$\Delta_{\rme,\bme}$ be the operator defined on symmetric 2-tensors by 
\begin{equation}\label{eq:Del:rme:bme}
\Delta_{\rme,\bme} h_{ij} = \frac{1}{\mu_{\rme}} \bnabla_m ( \rme^{mn} \mu_{\rme}
\bnabla_n h_{ij})
\end{equation}
see \cite[(1.8)]{AML}. 
Define the operator $P_{\rme,\bme}$ by
$$
P_{\rme,\bme} X^i = \rDelta[g] X^i + R[\rme]^i_j X^j - \Lie_X V^i_{\rme,\bme} -
2(\nabla[\rme]^m X^n) (\Gamma[\rme]^i_{mn} - \Gamma[\bme]^i_{mn})
$$
The defining equation for $X$ is given in terms of $P_{\rme,\bme}$ in equation
(\ref{eq:shiftresc2}) below. 
Note that we are interested here in the case where $V = 0$, but we include it
in the above formula since it makes the calculations below more transparent. 
Similarly, let $P_{\bme,\bme}$ be the corresponding operator with $g$
replaced by $\bme$ 
$$
P_{\bme,\bme} X^i = \bDelta X^i + \bR^i_j X^j - \Lie_X V^i_{\bme,\bme} -
2(\bnabla^m X^n) (\bGamma^i_{mn} - \bGamma^i_{mn})
$$
where the last two terms vanish identically. 
By the results of \cite[\S 5]{AML}, these operators take the form 
\begin{align} 
P_{\rme,\bme} X^i &= \rme^{mn} \bnabla_m \bnabla_n X^i + \rme^{mn}
\bR^i{}_{mjn}X^j  \label{eq:Prmebme} \\
P_{\bme,\bme} X^i &= \bme^{mn} \bnabla_m \bnabla_n X^i + \bme^{mn}
\bR^i{}_{mjn} X^j \label{eq:Pbmebme}
\end{align} 
where the index on $\bR$ is raised with $\bme$. By the discussion in section 
\ref{sec:HC}, $P_{\rme,\bme}: H^{s+1} \to H^{s-1}$ is an isomorphism, as is
$P_{\bme,\bme}$. From this it is straightforward to show
 that the inequality 
$$
||\rShift||_{H^{s+1}} \leq C ( ||\rKtr||_{H^{s-1}} + ||\rme - \bme||_{H^s} ) 
$$
holds, i.e. we have proved (\ref{eq:rS-est}). Here we have made use of (\ref{eq:lastterm-est}) 
and absorbed terms which can be estimated by 
$$
(||\rme-\bme||_{H^s} + ||\bme - \bme_0||_{H^s} ) ||\rShift||_{H^{s+1}}
$$ 
after a small change in the constant.

It remains to prove (\ref{eq:rSY-est}). 
We  write the defining equation for $\rShift$ as 
\begin{align} 
P_{\rme,\bme} \rShift^i &= 
-2\rLapse \rKtr^{mn}  (\rGamma[g]_{mn}^i 
- \bGamma_{mn}^i) 
- 2( \frac{\rLapse}{n}-1) V^i \nonumber \\
&\quad + 2 (\rnabla^m \rLapse )\rKtr^i_{\ m} 
+ (2 - n) \rnabla^i
( \frac{\rLapse}{n}-1) \nonumber 
\\ 
&\quad - \rme^{mn} \partial_T \bGamma_{mn}^i \label{eq:shiftresc2} 
\end{align} 
where by the gauge conditions we may set $V=0$. 
All terms in the right hand
side of (\ref{eq:shiftresc2}) are quadratic in small quantities, 
except the last. 

Recall that $\partial_T\bme$ is of the form 
(\ref{eq:dTbme-form}). 
A calculation shows 
$$
\bme^{mn} D\bGamma^i_{mn} . h^{\TTpara} = 0
$$
due to the fact that $h^{\TTpara}$ is transverse traceless 
with respect to $\bme$. 
Therefore we have 
$$
\bme^{mn} \partial_T \bGamma^i_{mn} = \bme^{mn} D\bGamma^i_{mn}
.(\Lie_{Y^\parallel} \bme)
$$
and a direct calculation gives 
$$
\bme^{mn} \partial_T \bGamma^i_{mn} = \bDelta Y^{\parallel\, i} + \bR^i_f
Y^{\parallel\, f} 
$$
Comparing with (\ref{eq:Pbmebme}) we have 
$$
\bme^{mn} \partial_T \bGamma^i_{mn} = P_{\bme,\bme} Y^{\parallel\, i} 
$$
and hence we may write the defining equation for $\rShift$ in the form 
\begin{align*} 
P_{\bme,\bme} (\rShift^i + Y^{\parallel\, i}) &= 
-2\rLapse \rKtr^{mn}  (\rGamma[g]_{mn}^i 
- \bGamma_{mn}^i) \nonumber \\
&\quad + 2 (\rnabla^m \rLapse) \rKtr^i_{\ m} 
+ (2 - n) \rnabla^i
( \frac{\rLapse}{n}-1) \nonumber \\
&\quad 
- (\rme^{mn} - \bme^{mn} ) (\bnabla_m \bnabla_n \rShift^i + \bR^i{}_{mjn}
\rShift^j) \nonumber \\ 
&\quad - (\rme^{mn} - \bme^{mn}) \partial_T \bGamma_{mn}^i 
\label{eq:shiftresc3} 
\end{align*} 
where we have set $V =0$. Estimating each term gives, after making use 
(\ref{eq:hLapse-est}) and (\ref{eq:rS-est}), as well as of the elementary 
inequality $ab \leq \half (a^2 + b^2)$ the estimate (\ref{eq:rSY-est}). 
\end{proof} 

\subsection{Splitting the Einstein equations} \label{sec:splitEin}
In this section, we will in preparation to proving the energy estimates
needed for the proof of our main result, rewrite the scale invariant Einstein
along the lines of \cite{AML}. We shall need the following estimate for the
curvature term in (\ref{eq:evolution-tKtr}).
Let
$\LL_{\rme,\bme}$ be the operator on symmetric 2-tensors defined by 
\begin{equation}\label{eq:LLrme-def}
\LL_{\rme,\bme} h
  = - \Delta_{\rme,\bme} h -2 \cR_{\bme} h
\end{equation} 
Here $\cR_{\bme} h$ is given by (\ref{eq:cRdef}) with the curvature tensor
$\bR$. In particular, 
if $\rme=\bme$ then $\LL_{\rme,\bme}$ coincides with the operator $\LL$
defined by (\ref{eq:LLdef}). 
The following Lemma follows from the form of $R_{ij}$ derived in the proof of
\cite[Theorem 3.1]{AML}.
\begin{lemma} \label{lem:Rij-form} 
$$
R_{ij} - \del_{ij} + \frac{(n-1)}{n^2} \rme_{ij} = \half \LL_{\rme,\bme} (\rme-\bme)_{ij} +
J_{ij}
$$
where 
$$
||J||_{H^{s-1}} \leq C ||\rme-\bme||_{H^s}^2 
$$
\end{lemma} 
We will write the Einstein evolution equations in terms of variables $(u,v)$
defined by 
\begin{subequations}\label{eq:uvalpha-def}
\begin{align}
u &= \rme-\bme, \quad v = 2n\rKtr \\
\intertext{and the normalized lapse}
\nLapse &= \frac{\rLapse}{n} 
\end{align}
\end{subequations} 
We use the identity 
\begin{equation}\label{eq:Lie-exp}
\Lie_X u_{ij} = X^m \bnabla_m u_{ij} + u_{im}\bnabla_j X^m + u_{mj} \bnabla_i
X^m
\end{equation} 
to expand the Lie derivative. With these definitions, we have 
\begin{lemma} \label{lem:dTuv}
The Einstein evolution equations (\ref{eq:Einrescaled}) are equivalent to the
system 
\begin{subequations} \label{eq:dTuv} 
\begin{align} 
\partial_T u &= \nLapse v - h^{\TTpara} - \rShift^i \bnabla_i u + \FF_u
\label{eq:dTu} \\
\partial_T v &= -(n-1) v - n^2\nLapse \LL_{\rme,\bme} u - \rShift^i \bnabla_i v + \FF_v
\label{eq:dTv} 
\end{align} 
\end{subequations} 
where 
\begin{subequations} \label{eq:FF-est}
\begin{align} 
||\FF_u||_{H^s} &\leq C (||u||_{H^s}^2 + ||v||_{H^{s-1}}^2 )  \label{eq:FFu-est}\\
||\FF_v||_{H^{s-1}} &\leq C (||u||_{H^s}^2 + ||v||_{H^{s-1}}^2 ) \label{eq:FFv-est} 
\end{align} 
\end{subequations} 
\end{lemma} 
\begin{proof} 
A direct calculation gives 
\begin{align*} 
\partial_T u &= \nLapse v - h^{\TTpara} + 2(\nLapse - 1) \rme -
\Lie_{\rShift} \rme - \Lie_{Y^\parallel} \bme \\
&= \nLapse v - h^{\TTpara} - \Lie_{\rShift} u \\
&\quad  + 2(\nLapse - 1) \rme - \Lie_{\rShift + Y^{\parallel}} \bme
\end{align*} 
We now expand the term $\Lie_{\rShift} u$ using (\ref{eq:Lie-exp}) and put
$\partial_T u$ in the form (\ref{eq:dTu}). The term $\FF_u$ defined in this
manner can be shown, using the inequalities of 
Lemma \ref{lem:quad-NX-est}, to satisfy the estimate (\ref{eq:FFu-est}).  
 
For $v$ we proceed in a similar manner, making use of Lemma
\ref{lem:Rij-form} to rewrite the curvature term in equation
(\ref{eq:evolution-tKtr}). 
\end{proof}

\subsection{The linearized Einstein equations} \label{sec:linearized} 
It is straightforward to linearize the system defined by equations
(\ref{eq:rgauge}), (\ref{eq:Conrescaled})--(\ref{eq:shadowdef}) and
(\ref{eq:dTuv}) about the exact solution $(\rme, \rKtr, \rLapse, \rShift) =
(\bme, 0 , n, 0)$ where $\bme$ lies in the deformation space $\NN$ of a fixed
Einstein metric $\bme_0$. Linearization of the Hamiltonian constraint 
(\ref{eq:Conrescaled-ham}) and the gauge condition (\ref{eq:rgauge}),
together with the Einstein condition satisfied by $\bme$, immediately imply
that the first variation $\delta u$ of $u = \rme - \bme$ is transverse
traceless with respect to the background metric $\bme$. Linearization of the
momentum constraint (\ref{eq:Conrescaled-mom}) together with the condition
$\rme^{ij} \rKtr_{ij} = 0$ implies that $\delta v_{ij} = 2n \delta
\rKtr_{ij}$ is also transverse traceless with respect to $\bme$. Variation of
equations (\ref{eq:defining-resc}) leads to $\delta \rLapse = 0$ and $\delta
\rShift^i + \delta Y^{\parallel\, i} = 0$. 
From equation (\ref{eq:PYi}) one finds that $\delta Y^\parallel$ is
determined from $\delta h^{\TTpara}$ which latter is also transverse
traceless with respect to $\bme$ and satisfies $\Ppara \delta h^{\TTpara} =
\delta h^{\TTpara}$. 
Variation of the shadow metric condition (\ref{eq:shadowdef}) shows further
that $\Ppara (\delta u) = 0$, i.e., that the transverse traceless tensor
$\delta u$ satisfies $\delta u = \delta u^{\perp}$. 

It is now straightforward to linearize the evolution equations
(\ref{eq:dTuv}) and decompose them into $\parallel$ and $\perp$
projections. This leads immediately to $\delta v^{\parallel} = \delta
h^{\TTpara}$ and to 
\begin{align*} 
\partial_T \delta u^\perp &= \delta v^\perp \\ 
\partial_T \delta v^{\parallel} &= -(n-1) \delta v^{\parallel} \\ 
\partial_T \delta v^{\perp} &= - (n-1) \delta v^{\perp} - n^2 \LL \delta
u^{\perp}
\end{align*} 
where $\LL$ is the operator given by (\ref{eq:LLrme-def}) with
$\rme=\bme$. These combine to give the second order euqation 
\begin{equation}\label{eq:uperp-evol}
\delta u^\perp_{,TT} + (n-1) \delta u^\perp_{,T} + n^2 \LL \delta u^\perp = 0
\end{equation}
for $\delta u^\perp$ and to give immediately that 
$$
\delta v^{\parallel} = \delta v^{\parallel} \big{|}_{T = T_0}
e^{-(n-1)(T-T_0)} \, .
$$
It follows that $\delta h^{\TTpara} = \delta v^\parallel$ and therefore also
$\delta Y^\parallel$, all decay at the same universal exponential rate (at
least in this linearized approximation). 

While equation (\ref{eq:uperp-evol}) may be solved explicitly by separation
of variables (as shown in \cite{fischer:moncrief:n+1} and recalled below) we
shall need to prove energy estimates for this system in order to have a tool
adequate for generalization to the nonlinear problem. In any case, recalling
that we can write 
$$
\delta \bme = \delta q^{(\alpha)} (h^{(\alpha) \, \TTpara} +
\Lie_{Y^{(\alpha)\, \parallel}} \bme ) 
$$
where 
$h^{(\alpha) \, \TTpara}$ and $Y^{(\alpha)\, \parallel}$ are background
quantities and $\delta q^{(\alpha)} = \delta q^{(\alpha)}(T)$ are function of
$T$ only, we can decompose $\delta u = \delta \rme - \delta \bme$ into its
constituents and show that 
$$
\delta q^{(\alpha)}(T) = - \left ( 
\partial_T (\delta q^{(\alpha)} ) \big{|}_{T=T_0} \right ) 
\frac{1}{n-1} 
e^{-(n-1)(T-T_0)} + \delta q^{(\alpha)} (\infty) 
$$
where $\delta q^{(\alpha)}(\infty)$ is a constant of integration that
yields the asymptotic value of $\delta \bme$. The perturbed metric $\delta
\rme$ is thus given by $\delta \rme = \delta u^\perp + \delta \bme$, with
$\delta u^\perp$ and $\delta \bme$ determined as above. 

In order to understand how to prove energy estimates for the system
(\ref{eq:uperp-evol}) 
we 
perform a separation of variables. 
Let $\lambda$ be a nonzero eigenvalue $\LL$ and let $X$ be an eigentensor
corresponding to $\lambda$. 
Equation (\ref{eq:uperp-evol}) gives the model system 
\begin{equation}\label{eq:model}
\ddot X + (n-1) \dot X + n^2 \lambda X = 0,
\end{equation}
which we recognize as a damped oscillator equation. 
This has characteristic roots
$$
\frac{-(n-1) \pm \sqrt{(n-1)^2- 4n^2\lambda}}{2}
$$
If 
$\lambda > \frac{(n-1)^2}{4n^2}$, then the characteristic equation has a
complex pair of roots with real part $-(n-1)/2$, and hence there is a
universal exponential rate of decay $-(n-1)/2$. 
If $0 < \lambda < \frac{(n-1)^2}{4n^2}$ the
characteristic equation has a pair of negative real roots. 
In this case we have an ``anomalous'' rate of decay
depending on $\lambda$. 
If $\lambda = \frac{(n-1)^2}{4n^2}$ we have a 
critically damped oscillator. We avoid dealing directly with this case, by 
decreasing $\lambda$ slightly.

\subsection{Energies for the damped oscillator} \label{sec:damped} 
The equation (\ref{eq:model}) is an ODE with constant coefficients 
and can therefore be analyzed by
elementary means. However, since the 
analysis of this system plays a central role in this paper we present a
complete derivation of an energy estimate, which will be
used later on to prove that energy for the rescaled Einstein
equations has the decay property needed for the main result of this paper. 

In
this section we consider the situation that $\lambda \geq \lambda_0$ for
some $\lambda_0> 0$,  
$\lambda_0 \ne \frac{(n-1)^2}{4n^2}$. Let $-\alpha_+$
denote the real part of 
$$
\frac{-(n-1) + \sqrt{(n-1)^2- 4n^2\lambda_0}}{2} ,
$$
i.e. $\alpha_+ = \alpha_+(n,\lambda_0)$ is given by 
$$
\alpha_+ = \left\{ \begin{array}{cc} \frac{n-1}{2} , & \lambda_0 >
\frac{(n-1)^2}{4n^2} \\
\frac{(n-1) - \sqrt{(n-1)^2- 4n^2\lambda_0}}{2},  & 0 < \lambda_0 <
\frac{(n-1)^2}{4n^2}  
\end{array} \right. 
$$
Define the constant $\Cenerg = \Cenerg(n,\lambda_0)$ by 
$$
\Cenerg = \left\{ \begin{array}{cc} \frac{n-1}{2}, & \lambda_0 >
\frac{(n-1)^2}{4n^2} \\
\frac{2n^2\lambda_0}{(n-1)}, & 0 < \lambda_0 < \frac{(n-1)^2}{4n^2} 
\end{array} \right. 
$$
Define the energy $E = E(X,\dot X;n,\lambda,\lambda_0)$ by 
$$
E = \half \dot X^2 + \frac{n^2 \lambda}{2} X^2 + \Cenerg X \dot X
$$
and let $X$ be a solution to the damped oscillator equation (\ref{eq:model})
for some $\lambda \geq \lambda_0$. 

\begin{lemma} \label{lem:Edefinite} 
The energy $E$ is positive definite for $\lambda_0 > 0$. Assume $\lambda \geq
\lambda_0 >
0$ and $\lambda_0 \ne \frac{(n-1)^2}{4n^2}$. 
Then $E$ satisfies 
$E$ satisfies $\dot E \leq -2\alpha_+ E$. 
\end{lemma} 
\begin{proof} 
We will consider the case $0 < \lambda_0 <
\frac{(n-1)^2}{4n^2}$ where anomalous decay holds. The 
case $\lambda_0 > \frac{(n-1)^2}{4n^2}$ with a universal rate of decay
is straightforward and will be left
to the reader. The energy $E$ corresponds to the quadratic form 
$$
E = \half \begin{pmatrix} 1 & \Cenerg \\
\Cenerg & n^2\lam \end{pmatrix} 
$$
Let $A = 2E$. 
In the anomalous case, $\Cenerg = 2n^2\lambda_0/(n-1)$, and setting 
$Y = 4n^2\lambda_0/(n-1)^2 $, we have $0 < Y < 1$. Then 
\begin{align*} 
\tr A &= 1 + n^2\lam \\
\det A & \geq n^2\lambda_0 (1 - Y)
\end{align*} 
Therefore we have $\tr A > 0, \det A > 0$ and 
it follows that $E$ is positive definite. 

A calculation shows 
$$
\dot E = - 2\alpha_+ E + J
$$
with 
\begin{multline}\label{eq:Jeq}
J = (\Cenerg - (n-1) + \alpha_+) \dot X^2 
\\ - (\Cenerg - \alpha_+ )n^2 \lambda X^2 
-\Cenerg ((n-1)-2\alpha_+)  X \dot X .
\end{multline}
which corresponds to the quadratic form 
$$
\begin{pmatrix} \Cenerg -(n-1)+\alpha_+ & -\Cenerg ((n-1)-2\alpha_+)/2  \\
-\Cenerg ((n-1)-2\alpha_+)/2 & -(\Cenerg -\alpha_+)n^2\lambda \end{pmatrix} 
$$
$J$ is in the anomalous case of the form 
$$
J = \frac{n-1}{2} \begin{pmatrix} Y - 1 - \sqrt{1-Y} & 
- \frac{2n^2 \lambda_0}{n-1} \sqrt{1-Y} \\
- \frac{2n^2 \lambda_0}{n-1} \sqrt{1-Y} & -n^2\lambda [Y-1 +\sqrt{1-Y}] 
\end{pmatrix}
$$
Setting $B = \frac{2}{n-1} J$, the determinant and trace of $B$ are given by 
\begin{align*} 
\det B &= (n^2\lambda - n^2 \lambda_0)Y(1-Y) \\
\tr B &=  (Y-1-\sqrt{1-Y}) - n^2 \lambda \sqrt{1-Y}(1-\sqrt{1-Y})
\end{align*} 
From this we see that $\tr B < 0$, $\det B > 0$.
Thus, 
the quadratic form $B$ and hence also $J$ is
negative definite. It follows that 
$$
\dot E \leq -2\alpha_+ E
$$ 
as claimed. 
\end{proof}

\section{Energy estimate} 
\label{sec:energy} 
Taking the analysis in section \ref{sec:damped} as a guide, we will now
define energies for the full Einstein equations. Let the operator
$\LL_{\rme,\bme}$ be given by (\ref{eq:LLrme-def}). Recall that 
$\Delta_{\rme,\bme}$ 
as defined in (\ref{eq:Del:rme:bme})
is the rough 
Laplacian on a certain vector bundle $Q$ over $(M,\rme)$, see 
\cite[\S 2]{AML} for discussion, which can be identified with the bundle of
symmetric covariant 2-tensors on $(M,g)$ with covariant derivative
$\bnabla$ and fiber metric defined in terms of $\bme$ by 
$$
\la u, v \ra = u_{ij} v_{kl} \bme^{ik} \bme^{jl} . 
$$
The corresponding norm is $|u| = (\la u , u \ra)^{1/2}$. It follows from the
definition that 
the covariant derivative $\bnabla$ is metric with respect to $\la \cdot,
\cdot \ra$. 
The inner product on derivatives is 
$$
\la \bnabla u , \bnabla v \ra = \la \bnabla_m u , \bnabla_n v \ra \rme^{mn} .
$$
The rough Laplacian
$\Delta_{\rme,\bme}$ is formally self-adjoint with respect to the natural $L^2$
inner product 
$$
\int_M \la u , \Delta_{\rme,\bme} v \ra \mu_{\rme} = \int_M \la \Delta_{\rme,\bme} u , v
\ra \mu_{\rme} .
$$
It follows that the operator $\LL_{\rme,\bme}$ is self-adjoint with respect to this inner
product. We are now able to define the energies which will be used for the
Einstein equations. The energy for the damped oscillator consists of a
standard oscillator energy and a correction term. Analogously, 
the energies we are about to define for the rescaled Einstein equations 
will consist of a wave equation type energy and a
correction term. To connect with the damped oscillator energy we use the
correspondence $X \leftrightarrow u$ and 
$\dot X \leftrightarrow v$, where $u=\rme-\bme,v=2n\rKtr$,
cf. (\ref{eq:uvalpha-def}). 

Throughout the rest of the paper, we fix $s > n/2 + 1$, assume
that the triple $(\bme,\rme,\rKtr)$ satisfies the shadow metric condition,
and that $(\rme, \rKtr) \in B_{\delta,s}(\bme,0)$ for some $\delta > 0$
suffiently small. 
 
The first order energy and correction term is 
\begin{align*} 
\EE{1} &= \half  \int_M |v|^2 \mu_{\rme}  
+ \half n^2 \int_M \la u, \LL_{\rme,\bme} u \ra \mu_{\rme} \\
\GG{1} &= \int_M \la v , u \ra \mu_{\rme} 
\end{align*} 
Explicitly, substituting in $\rme-\bme,\rKtr$ using the relation 
$u=\rme-\bme$, $v=2n\rKtr$ this gives
\begin{align*} 
\EE{1} &= \half (2n)^2 \int_M \rKtr_{ij} \rKtr_{kl} \bme^{ik} \bme^{jl}
\mu_{\rme} \\
&\quad 
+ \half n^2 \int_M \left\{ \bnabla_k (\rme_{ij} - \bme_{ij} )
\bnabla_l (\rme_{mn} - \bme_{mn} ) \rme^{kl} \bme^{im} \bme^{jn} \right. \\
&\quad \left. 
-2 \bR_{i\ j}^{\ k\ l} (\rme_{kl} - \bme_{kl} ) 
(\rme_{mn} - \bme_{mn})\bme^{im}\bme^{jn} \right\} \mu_{\rme} 
\end{align*} 
The correction term
$\GG{1}$ can be expanded in a similar manner. We now define higher order
energies by inserting suitable powers of $\LL_{\rme,\bme}$, 
giving for integers $m \geq 1$,
\begin{align*} 
\EE{m} &= \half  \int_M \la v , \LL_{\rme,\bme}^{m-1} v \ra \mu_{\rme} 
+ \half n^2 \int_M \la u , \LL_{\rme,\bme}^{m} u \ra \mu_{\rme} 
\\
\GG{m} &=  \int_M \la v , \LL_{\rme,\bme}^{m-1} u \ra \mu_{\rme} 
\end{align*} 
Due to the shadow metric condition, $\bme$ may be viewed as a function of
$\rme$, and hence the energies $\EE{s}$ depend only on $(\rme, \rKtr)$. 
In case the lowest eigenvalue of $\LL_{\bme_0,\bme_0}$ 
is zero at the initial
background metric, let $\lamminmod > 0$ be the smallest nonzero eigenvalue of
$\LL_{\bme_0,\bme_0}$ 
and let 
\begin{equation}\label{eq:lambda0def}
\lambda_0 =\lamminmod - \eps
\end{equation}
for some $\eps > 0$. 
We require $\lambda_0 > 0$. The reason for choosing
$\lambda_0$ smaller than $\lamminmod$ is that the spectrum of $\LL_{\bme(T)}$
depends on $T$ and it is necessary that the energy estimates we shall prove
hold uniformly during the course of the evolution. 

Let now $\Cenerg, \alpha_+$ be defined as in section \ref{sec:damped} 
in terms of the $\lambda_0$ chosen
above. 
For integers $m \geq 1$, let 
$$
\EEtot{(m)} = \EE{m} + \Cenerg \,\GG{m} 
\,.
$$
Based on the work in section \ref{sec:damped}, one expects that a
corrected energy of the form 
\begin{equation}\label{eq:EEtotdef}
\EEtot{s} = \sum_{1 \leq m \leq s} \EEtot{(m)}
\end{equation} 
will have the property that 
$$
\partial_T \EEtot{s} \leq -2\alpha_+ \EEtot{s} +
\text{higher order terms} 
$$
This is indeed the case, as will be shown below.

\begin{remark} The value of $\lambda_0$ determines $\alpha_+$ and hence the
  decay rate that is proved by the present argument. 
A more detailed
  analysis, along the lines of \cite{ChB:moncrief:U(1)}, can be used to prove
  a sharp decay estimate.
\end{remark} 
\subsection{Positive definiteness of the energy} \label{sec:hessian} 
Recall that $\Pperp_{\bme}$ which was introduced in section \ref{sec:linein}
is the $L^2$-orthogonal projection onto the orthogonal complement of $\ker
\LL$ in the space of $\TT$ tensors with respect to $\bme$. 
It is clear from the construction
that for $s > n/2+1$, $\EEtot{s}$ is a smooth function on $\CC^s$ and
further that $\Pperp_{\bme}$ depends smoothly on the Einstein metric
$\bme$. 
\begin{lemma} \label{lem:D2E}
Let $\bme$ be an
Einstein metric on $M$ with Einstein constant
$-(n-1)/n^2$ and let $\EEtot{s}$ be the total energy defined in section
\ref{sec:energy} with $s > n/2 +1$. 
Then there is a $\delta > 0$
and a constant $C > 0$, such that for $(\rme, \rKtr) \in
\BB_{s,\delta}(\bme,0)$, the inequality 
\begin{equation} \label{eq:EEtot-ineq}
||\Pperp_{\bme}(\rme - \bme)||_{H^{s}}^2 + ||\rKtr||_{H^{s-1}}^2 \leq C \EEtot{s}
\end{equation} 
holds. 
\end{lemma} 
\begin{proof} 
Note that $(\bme,0)\in \CC^{s}$ is a critical point of
$\EEtot{s}$. Therefore it suffices to consider the second derivative 
of the energy at
$(\bme, 0)$. Let $m$ be an integer such that $1 \leq m \leq s$. 
The Hessian of 
$\EEtot{(m)}$ is of the form 
\begin{multline*}
D^2 \EEtot{(m)}((h,k),(h,k)) = \\
\int_M \la k , \LL_{\bme,\bme}^{m-1} k \ra  \, \mu_\bme
+ n^2 \int_M \la h, \LL_{\bme,\bme}^m h\ra \, \mu_\bme + 
2 \Cenerg \int_M \la k, \LL_{\bme,\bme}^{m-1} h \ra \, \mu_\bme \,.
\end{multline*}
An analysis using the spectral decomposition of
$\LL_{\bme,\bme}$ shows, using the arguments in the proof of 
Lemma \ref{lem:Edefinite} that $D^2 \EEtot{(m)}$
satisfies 
$$
D^2 \EEtot{(m)}.((h,k),(h,k)) \geq 0
$$ 
with equality if and only if 
$(h,k) = (h^{\TT\,\parallel},0)$ 
with $h^{\TTpara} \in \ker \LL_{\bme,\bme}$. From this follows
that   
$$
|| \Pperp_{\bme} h ||_{H^{s}}^2 + ||k||_{H^{s-1}}^2 \leq C D^2
\EEtot{s}.((h,k),(h,k)) 
$$
for some constant $C = C(\lambda_0,\bme)> 0$, where $\lambda_0$ is defined in
(\ref{eq:lambda0def}). It follows from the above and 
Taylor's theorem that there is
a $\delta > 0$ such that the inequality (\ref{eq:EEtot-ineq}) holds in 
$\BB_{s,\delta}(\bme,0)$ 
for suitable $\delta > 0$, $C > 0$. 
\end{proof} 

\begin{lemma} \label{lem:projperp}
Let $(\bme,\rme,\rKtr)$ be as in Lemma \ref{lem:D2E}. 
There is a $\delta > 0$
sufficiently small, and a constant $C > 0$ so that if $||\rme -
\bme||_{H^{s}} \leq \delta$, 
$$
||\Ppara_{\bme} (\rme - \bme) ||_{H^s} 
\leq C ( ||\Pperp_{\bme}(\rme - \bme)||_{H^s}^2 + || \rKtr ||_{H^{s-1}}^2 )
\, .
$$
\end{lemma} 
\begin{proof} 
By the analysis in section \ref{sec:slice}, we may write 
\begin{equation}\label{eq:uvgraph} 
\rme - \bme = u^{\TT}  + z, \quad \rKtr = v^{\TT} + w 
\end{equation}
where $u^{\TT}, v^{\TT}$ are $\TT$ tensors with respect to
$\bme$ and $z,w$ are $L^2$ perpendicular to the space of $\TT$ tensors and
satisfy 
$$
||z||_{H^s} + ||w||_{H^{s-1}} \leq C ( ||u^{\TT}||_{H^s}^2 +
||v^{\TT} ||_{H^{s-1}}^2 )
$$ 
for $\rme$ sufficiently close to $\bme$. 
Recall that for any $\TT$ tensor with respect to $\bme$, 
$$
(u^{\TT}, \Lie_Y \bme)_{L^2} = 0 .
$$ 
Using $\Pperp_{\bme}$ we may split 
$u^{\TT}$ $L^2$-orthogonally as $u^{\TTpara} + u^{\TTperp}$. 
Taking equation (\ref{eq:halpha-first}), 
and the just mentioned facts into account 
one finds that equation (\ref{eq:shadowdef}) is equivalent to the
set of $m$ conditions 
\begin{align*} 
0 &= (u^{\TT} + z, h^{(\alpha) \, \TTpara} + \Lie_{Y^{(\alpha)\parallel}}
\bme)_{L^2} \\
&= (u^{\TTpara} , h^{(\alpha) \, \TTpara} )_{L^2} + 
(z, \Lie_{Y^{(\alpha)\parallel}}\bme)_{L^2}, \quad \alpha = 1, \dots , m. 
\end{align*} 
It follows that this relation defines $u^{\TTpara}$ as a smooth
function of $z$, vanishing at $z=0$, and hence in view of  
(\ref{eq:uvgraph}) $z$ is seen to be a function of $u^{\TTperp}, v^{\TT}$. 
Since
$z$ is of at least second order in $u^{\TT}, v^{\TT}$ 
and if the shadow relation
holds, of $u^{\TTperp}$, we have 
$$
||z||_{H^s} + ||w||_{H^{s-1}} \leq C ( ||u^{\TTperp}||_{H^s}^2 + ||v^{\TT}||_{H^{s-1}}^2 ) .
$$
The result follows. 
\end{proof} 

The following result is a direct consequence of Lemmas \ref{lem:D2E} and
\ref{lem:projperp}, and their proofs. 
\begin{thm}\label{thm:Ebound} 
Suppose that $(\bme, \rme, \rKtr)$ satisfy the shadow metric condition, and
let 
$\EEtot{s}$ be the total energy defined in section
\ref{sec:energy} for $s > n/2+1$.
Then there is a $\delta > 0$
and a constant $C > 0$, such that for $(\rme, \rKtr) \in
\BB_{s,\delta}(\bme,0)$, 
the inequality 
$$
||\rme - \bme||_{H^{s}}^2 + ||\rKtr||_{H^{s-1}}^2 \leq C \EEtot{s}
$$
holds. 
\end{thm} 

We are now able to state the following version of
the continuation principle.
\begin{cor} \label{cor:continuation} 
Let $(\bme_0, \rme_0, \rKtr_0)$ be an initial data set as in Theorem
\ref{thm:Ebound}, at an initial time $T_0$. 
Let $[T_0, T_+)$ be the maximal existence interval in $H^s$, $s > n/2 + 1$, 
for the rescaled Einstein
  equations   with the shadow metric condition imposed, 
 with initial data $(\bme_0, \rme_0, \rKtr_0)$. 

Then there are numbers $\delta_0 > 0$, 
$\delta > 0$ so that if
$(\rme(T_0),\rKtr(T_0)) \in \BB_{s,\delta_0}(\bme_0,0)$, 
satisfies
$E_s(\rme(T_0),\rKtr(T_0)) < \delta$, then either
$T_+ = \infty$ or there is a finite time $T < \infty$ such that 
$$
E_s(T) \geq \delta .
$$
\end{cor} 
This result reduces the problem of proving global existence for the system 
(\ref{eq:Einrescaled}, \ref{eq:Conrescaled}, \ref{eq:defining-resc}) with
the shadow metric condition (\ref{eq:shadowdef}) imposed to the problem of
proving that the energy $E_s$ stays small, if it is small initially. 

\subsection{Time derivative of the energy} 
We now consider the time derivative of the energy. In order to see the
pattern, we do the calculation for the first order energy separately. 
In the situation considered in Theorem \ref{thm:Ebound} and Corollary
\ref{cor:continuation}, we have 
\begin{lemma} 
Suppose that $(\bme, \rme, \rKtr)$ satisfy the shadow metric condition.
There is a $\delta > 0$
such that for $(\rme, \rKtr) \in
\BB_{s,\delta}(\bme,0)$, $s > n/2 + 1$, 
we have 
$$
\partial_T \EE{1} =  - (n-1) \int_M |v|^2 \mu_{\rme}  + U_1
$$
and 
$$
\partial_T \GG{1} \leq \int_M \la ( -(n-1) v, u \ra + |v|^2 ) \mu_{\rme} 
- n^2 \int_M \la \nLapse  \LL_{\rme,\bme} u,u \ra \mu_{\rme} + V_1 
$$
where 
$$
|U_1| + |V_1| \leq C ( ||\rme-\bme||_{H^s}^3 + ||\rKtr||_{H^{s-1}}^3 ) 
$$
\end{lemma} 
\begin{proof} 
We have 
\begin{align*} 
\partial_T \EE{1} &= 
\int_M \la v, -(n-1) v \ra \mu_{\rme} - n^2 \int_M \la
\LL_{\rme,\bme} u, h^{\TTpara} \ra \mu_{\rme} \\
&\quad - \int_M \la v , X^i \bnabla_i v \ra \mu_{\rme} - n^2 \int_M \la
\LL_{\rme,\bme} u , X^i \bnabla_i u \ra \mu_{\rme} + R_1
\end{align*} 
where $R_1$ is third order. In particular, using the estimates for
$\rShift$ 
we have 
$$
|R_1| \leq C ( ||u||_{H^s}^3 + ||v||_{H^{s-1}}^3) \, .
$$
Further, due to the self-adjointness of $\LL_{\rme,\bme}$ and the 
fact that $\LL_{\bme,\bme} h^{\TTpara} = 0$, we have that 
$$
\left |  \int_M \la
\LL_{\rme,\bme} u, h^{\TTpara} \ra \mu_{\rme} \right | \leq  C ( ||u||_{H^s}^3 +
||v||_{H^{s-1}}^3)
$$
The terms
\begin{align*} 
\int_M & \la v ,  X^i \bnabla_i v \ra \mu_{\rme} \\
\int_M & \la \LL_{\rme,\bme} u , X^i \bnabla_i u \ra \mu_{\rme}
\end{align*} 
also clearly 
satisfy a third order estimate.
Next we consider the correction term. We have, proceeding as above after a
direct calculation  
\begin{align*} 
\partial_T \GG{1} &= \int_M ( \la -(n-1) v , u \ra + |v|^2 ) \mu_{\rme} 
 - n^2 \int_M \la \nLapse \LL_{\rme,\bme} u , u \ra \mu_{\rme}  \\
&\quad + \int_M \la v, (\nLapse -1) v - h^{\TTpara} \ra \mu_{\rme} \\
&\quad - \int_M ( \la X^i \bnabla_i v , u\ra + \la v , X^i\bnabla_i u
\ra ) \mu_{\rme} + S_1
\end{align*} 
where $S_1$ is third order.

Time differentiating equation (\ref{eq:shadowdef}) gives, in view of
(\ref{eq:dTbme-form}),
\begin{align*}
0 &= 
\int_M (\rme_{ij,T} - \bme_{ij,T} ) \bme^{ik} \bme^{jl}
(h_{kl}^{(\alpha) \, \TTpara} + \Lie_{Y^{(\alpha)\, \parallel}} \bme_{kl})
\mu_\bme + \text{ second order terms} \\ 
&= 
\int_M (\rme_{ij,T} - h^{\TTpara}_{ij} - \Lie_{Y^\parallel} \bme_{ij} ) \bme^{ik} \bme^{jl}
(h_{kl}^{(\alpha) \, \TTpara} + \Lie_{Y^{(\alpha)\, \parallel}} \bme_{kl})
\mu_\bme + \text{ second order terms} 
\end{align*} 
Here we have made use of the fact that $\dot q^{(\alpha)}$ is of first order. 
By (\ref{eq:hLapse-est}), $\rLapse/n - 1$ is 
second order. 
This gives after using (\ref{eq:evolution-tg}) and simplifying,
\begin{align*} 
0 &= 
\int_M (2n \rKtr_{ij} - h^{\TTpara}_{ij} )  
\bme^{ik} \bme^{jl}
(h_{kl}^{(\alpha) \, \TTpara} + \Lie_{Y^{(\alpha)\, \parallel}} \bme_{kl})
\mu_\bme + \text{ second order terms} 
\end{align*}
We next note that modulo second order terms, 
$\rKtr$ is transverse traceless with respect to $\bme$. This gives 
\begin{align*} 
0 &= 
\int_M ( 2n \rKtr_{ij} - h^{\TTpara}_{ij} ) h^{(\alpha)\, \TTpara\, ij}
\mu_\bme
+  \text{ second order terms} \\
&= 
\int_M ( v_{ij} - h^{\TTpara}_{ij} ) h^{(\alpha)\, \TTpara\, ij}
\mu_\bme
+  \text{ second order terms} 
\end{align*} 
We have now proved that
\begin{equation}\label{eq:Pparav}
h^{\TTpara} = \Ppara v + \text{ second order terms} 
\end{equation} 
and hence 
$$
\int_M \la v, h^{\TTpara} \ra \mu_{\rme} = \int_M |v^{\TTpara}|^2 \mu_{\rme} 
+ \text{ third order terms} 
$$
Taking signs into account, we see that 
$$
\int_M \la v, (\nLapse -1) v - h^{\TTpara} \ra \mu_{\rme} 
= - \int_M |h^{\TTpara}|_\bme^2 \mu_\bme + \text{ third order terms}
$$
can be bounded from above by a third order term. The terms involving
$\rShift^i \bnabla_i$ can be handled as above. This completes the proof of the
Lemma. 
\end{proof} 
Since $s > n/2 + 1$ is assumed the standard product estimates,
cf. eg. \cite[section 2]{AML},  allow us to
use the fact that we are in a small data situation and handle 
the higher order terms in the energy in the same way. 
Let $\JJtot{s} = \sum_{1 \leq j \leq s} \JJ{j}$ where 
\begin{align*} 
\JJ{j} &= (\Cenerg - (n-1) + \alpha_+) \int_M \la \LL_{\rme,\bme}^{j-1} v,v \ra
\mu_{\rme} \\ 
&\quad - (\Cenerg - \alpha_+ )n^2 \int_M \la \LL_{\rme,\bme}^j u,u\ra
\mu_{\rme} \\
&\quad 
-\Cenerg ((n-1)-2\alpha_+)  \int_M \la \LL_{\rme,\bme}^{j-1} u, v \ra
\mu_{\rme} 
\end{align*} 
is defined in analogy with (\ref{eq:Jeq}). 
An analysis along the lines of Lemma \ref{lem:D2E}, using the estimate for
the term $J$ from the proof of Lemma \ref{lem:Edefinite}, shows that the
term $\JJtot{s}$ is nonpositive modulo a third order term. This gives 
\begin{lemma} 
Suppose that $(\bme, \rme, \rKtr)$ satisfy the shadow metric condition.
There is a $\delta > 0$
such that for $(\rme, \rKtr) \in
\BB_{s,\delta}(\bme,0)$, $s > n/2 + 1$, 
we have 
$$
\JJtot{s} \leq C ( ||\rme - \bme||_{H^s}^3 + ||\rKtr||_{H^{s-1}}^3) 
$$
\end{lemma} 

Putting these results together and using Theorem \ref{thm:Ebound} gives  
\begin{thm} \label{thm:energest} Suppose the assumptions of Corollary
  \ref{cor:continuation} hold. Then, after possibly decreasing $\delta$,
  there is a constant $C$ such that 
\begin{equation}\label{eq:Eineq}
\partial_T \EEtot{s} \leq - 2\alpha_+ \EEtot{s} + 2C \EEtot{s}{}^{3/2}
\end{equation} 
holds if $E_s < \delta$. 
\end{thm}

\section{Future complete spacetimes} \label{sec:global} 
In this section, we derive some consequences of the results 
we have proved for the rescaled
Einstein equations. 
Let $Y = E_s^{1/2}$ and write $\dot Y =
\partial_T Y$. Then (\ref{eq:Eineq}) takes the form 
$$
\dot Y \leq -\alpha_+ Y + C Y^2
$$
The model equation $\dot y = -\alpha_+ y + C y^2$ with $y(0) = y_0 > 0$  
has the solution 
$$
y = \frac{\alpha_+}{C + e^{\alpha_+(T-T_0)}[\alpha_+/y_0 - C]}
$$
which remains
bounded for $T \geq T_0$ if $y_0^{-1} > \frac{C}{\alpha_+}$. 
Since $\partial_T \rme$ and hence also $\partial_T \bme$ 
is bounded in terms of 
$E_s^{1/2}$,
we see that we can ensure that by starting sufficiently close to data of the
form $(\bme_0,0) \in \NN$ the conditions of Corollary \ref{cor:continuation}
remain satisfied for all $T \geq T_0$. 

We state the conclusion as 
\begin{thm}\label{thm:globalexist} 
Suppose $\bme_0$ has integrable deformation space. 
Then, there is a $\delta_1 > 0$ 
such that for any $(\rme_0, \rKtr_0) \in
\BB_{s,\delta_1}(\bme_0,0)$, with $(\bme_0, \rme_0, \rKtr_0)$ satisfying the
shadow metric condition,  
the Cauchy problem for the 
system (\ref{eq:Einrescaled},
\ref{eq:Conrescaled}, \ref{eq:defining-resc}), with
the shadow metric condition imposed, with initial data
$(\bme_0,\rme_0,\rKtr_0)$, is globally well-posed to the future.  
\end{thm} 
 
Next we consider the properties of the spacetimes corresponding to the
solutions constructed in Theorem \ref{thm:globalexist}. 
In order to do this, 
we rephrase the result in terms of the physical Cauchy data. 
Thus, let $(M,\pme_0,\pK_0)$ be CMC vacuum Cauchy data for the
Einstein equations with mean curvature $\tau_0 < 0$, and let the spacetime 
$(\aM, \ame)$ be the maximal Cauchy development
of 
$(M, \pme_0, \pK_0)$. 
Further, suppose that $(M,\pme_0, \pK_0)$ is such that
for the corresponding rescaled data 
$(\rme_0,\rKtr_0)$ 
at scale invariant time
$T_0$, there exists $\bme_0 \in \Ein_{-(n-1)/n^2}$ with integrable
deformation space $\NN$ and such that the triple $(\bme_0, \rme_0, \rKtr_0)$
is initial data for the rescaled Einstein equations, satisfying the shadow
metric condition, with $(\rme_0, \rKtr_0) \in B_{s,\delta_1}(\bme_0, 0)$ for
$\delta_1$ as in Theorem \ref{thm:globalexist}. The following Corollary to
Theorem \ref{thm:globalexist} follows by an argument along the lines of
\cite[\S 6.1]{AMF}. 
\begin{cor} 
Let $(\bme_0, \rme_0, \rKtr_0)$ be as in Theorem \ref{thm:globalexist}, let
$(\pme_0, \pK_0)$ be the corresponding physical Cauchy data, and
let $(\aM, \ame)$ be the maximal Cauchy development of $(\pme_0,
\pK_0)$. Then 
\begin{enumerate} 
\item 
The spacetime $(\aM, \ame)$ is future complete. 
\item 
The spacetime $(\aM, \ame)$ is globally foliated to the future of $(M,
  \pme,\pK)$ by CMC Cauchy surfaces, with mean curvature taking all values in
  $[\tau_0, 0)$. 
\end{enumerate} 
\end{cor} 
It follows from the
energy estimate that
\begin{equation}\label{eq:energy-decay}
||\rme - \bme  ||_{H^s} + ||\rKtr||_{H^{s-1}} \leq C
e^{-\alpha_+ T}
\end{equation}
as $T \to \infty$, for
some constant $C$. 
Recall that $\bme$ is the shadow metric of $\rme$. For the
solution curves $T \mapsto (\bme,\rme,\rKtr)$ to the rescaled Einstein
equations in CMCSH gauge considered in Theorem
\ref{thm:globalexist}, it holds by construction that $\bme$ stays in a
neighborhood of $\bme_0$ in $\NN$. In fact, it holds that $\bme(T)$ tends to
a limit in $\NN$ as $T \to \infty$. To see this we note the following.  
In view of (\ref{eq:dTbme-form}), we may estimate
$\partial_T \bme$ in terms of $h^{\TTpara}$ and $Y^\parallel$. Equation
(\ref{eq:PYi}) gives an estimate for $Y^\parallel$ in terms of
$h^{\TTpara}$. Further, equation (\ref{eq:Pparav}) allows us to estimate 
$h^{\TTpara}$ in terms of $v = 2n\rKtr$ up to terms which are of second order
in small quantities. Thus, the inequality (\ref{eq:energy-decay}) gives the
following corollary to Theorem \ref{thm:globalexist}. 
\begin{cor}
Let $(\bme_0, \rme_0, \rKtr_0)$ be as in Theorem \ref{thm:globalexist}, and
let $T \mapsto (\bme,\rme,\rKtr)$ be the maximal solution to the Cauchy problem 
for the 
system (\ref{eq:Einrescaled},
\ref{eq:Conrescaled}, \ref{eq:defining-resc}), with
the shadow metric condition imposed, with initial data
$(\bme_0,\rme_0,\rKtr_0)$. Then, there is $\bme_* \in \NN$ 
such that $(\bme, \rme,\rKtr) \to (\bme_*, \bme_*, 0)$ as $T \to
\infty$. 
\end{cor}
This shows that there is a limiting Einstein metric for the rescaled Einstein
flow, and hence motivates the title of the paper. The above result is
completely analogous to the results of 
\cite{andersson:etal:2+1grav,moncrief:teichmuller}, which
imply that in the case of the 2+1 dimensional vacuum Einstein equations, the
rescaled geometry for a CMC foliation converges to a point in 
Teichm\"uller space. Similarly, the conclusion in the work of Moncrief and
Choquet-Bruhat \cite{ChB:moncrief:U(1)}, 
is
that the conformal geometry of the CMC foliation of the base of the $\text{U}(1)$
bundle converges to a point in Teichm\"uller space. In the 
higher dimensional situation considered in the present paper, the Einstein moduli space
plays the same role as  Teichm\"uller space. 
 
\appendix

\section{The shadow gauge} \label{sec:shadow:gauge} 
In this section, we let $\NN$ be the deformation space with respect to
$\bme_0$ and assume that $\NN$ is integrable, of dimension $m$. 
We shall consider some aspects of the shadow gauge condition. We
make use of some standard facts from differential topology of infinite
dimensional manifolds. 
All spaces we shall deal with can be viewed either as smooth Hilbert manifolds,
modelled on Sobolev spaces $H^s$, or as Frechet manifolds, modelled
on $C^\infty$ viewed as a scale of Sobolev spaces. Further, all maps are
smooth with Fredholm Frechet derivatives. 

Let $V_{\rme,\bme}^i$ be the tension field, defined in (\ref{eq:Vdef}).  Let
$\calX$ be the space of vector fields on $M$. We can view $V$ as a map 
from the cartesian product of the space of metrics $\MM$ with the shadow manifold
$\NN$ to the space of vector fields, 
$$
V : \MM \times \NN \to \calX .
$$
The space of pairs $(\rme,\bme) \in \MM \times \NN$ such that $\rme$ is in
harmonic gauge with respect to $\bme$ is precisely the zero set of this
map. We now shift our attention to the constraint set $\CC$. 
The space $\CC$ is a smooth
submanifold of the space $T^{\tr}\MM$, consisting of 
pairs $(\rme,\rKtr)$ satisfying the constraint
equations (\ref{eq:Conrescaled}). 
By a slight
abuse of notation, we can view $V$ as a map $\CC \times \NN \to \calX$. 
Let $\Shade: \MM \to \NN$ be the shadow map, cf. remark
\ref{rem:shade}. This is a smooth map defined on a neighborhood of
$\NN$. Similarly to above, we can view $\Shade$ as defining a map $\CC \to
\NN$, defined locally near $(\bme_0,0)$. 

We shall now define a map $\hatV$ on a neigborhood of $\NN \times \{0\}
\subset \CC$, in terms of $V, \Shade$. We define 
$$
\hatV: \CC \to \calX, \quad \hatV: (\rme,\rKtr) \mapsto V_{\rme,\Shade(\rme)} .
$$
We now calculate $D\hatV \big{|}_{(\bme_0,0)}$. We have the decomposition 
$$
T_{(\bme_0,0)} \CC = (u^\TT + \Lie_Y \bme_0, v^\TT) .
$$
\begin{lemma} \label{lem:DhatV} 
$$
D \hatV \big{|}_{(\bme_0,0)} . (u^{\TT} + \Lie_Y \bme_0, v^\TT) =
P_{\bme_0,\bme_0} Y ,
$$
where the operator $P_{\bme_0,\bme_0}$ 
is as in (\ref{eq:Pdef}), defined at
  $(\rme,\bme) = (\bme_0, \bme_0)$. 
\end{lemma} 
\begin{proof} 
We consider the tension field $V$ as a map $(\rme,\bme) \to V_{\rme,\bme}$. 
Let $h = u^{\TT} + \Lie_Y \bme_0$, where $u^{\TT}$ is transverse traceless
with respect to $\bme_0$.  It follows from 
(\ref{eq:DV}) that $D_\rme V^i \big{|}_{\bme_0,\bme_0} . h =
P_{\bme_0,\bme_0} Y^i$. Using
this together with 
the fact that $D\Shade \big{|}_{\bme_0} . h = u^{\TTpara}$, 
as follows from equations (\ref{eq:dTbme-form}) and (\ref{eq:PYi}), 
we have  
\begin{align*} 
D\hatV^i \big{|}_{(\bme_0,0)}. (h, v^\TT) &=
D_g V^i \big{|}_{\bme_0,\bme_0} . h + 
D_\bme V^i \big{|}_{\bme_0,\bme_0} . D \Shade \big{|}_{\bme_0} . h \\
&= P_{\bme_0,\bme_0} Y^i + \bme_0^{mn} D\Gamma^i_{mn} \big{|}_{\bme_0}
u^{\TTpara} \\
&= P_{\bme_0,\bme_0} Y^i 
\end{align*} 
where in the last step we used 
(\ref{eq:DGammauTT}). 
\end{proof} 
As discussed in section \ref{sec:slice}, the operator $P_{\bme_0,\bme_0}$ is 
an isomorphism. This implies, by an application of the implicit function
theorem, the following corollary. Let $\Slice_{\CC,\NN}^\perp \subset \CC$ be the set of
solutions to $\hatV = 0$, and let $\Slice_{\CC,\bme}^\perp$ be the subset of
$\Slice_{\CC,\NN}^\perp$ consisting of $(\rme, \rKtr) \in \CC$ 
satisfying the shadow gauge
condition with respect to $\bme$. 
\begin{cor} 
There is a neighborhood $\UU$ of $(\bme_0,0) \in \CC$ such that 
$\hatV: \UU \to \calX$ 
is a submersion and 
$$
\UU \cap \Slice_{\CC,\NN}^\perp
$$
is a submanifold of $\UU$. 
\end{cor} 
Next we consider the action of the diffeomorphism group $\DD$.
In the applications to $\CC$,
it is sufficient to work in a neighborhood of $(\bme_0,0)$. 
\begin{lemma} \label{lem:TSlice} 
$$
T_{(\bme_0,0)} \Slice_{\CC,\NN}^\perp = T_{(\bme_0,0)} \Slice_{\CC,\bme_0} .
$$
\end{lemma} 
\begin{proof} 
Recall from section \ref{sec:slice} that 
$T_{(\bme_0,0)} \Slice_{\CC,\bme_0}$ is the space  
$\{ (u^\TT, v^\TT)\}$ where $u^\TT, v^\TT$ are 
$\TT$ tensors with respect to $\bme_0$. 
This shows, in view of the proof of lemma \ref{lem:DhatV},   
that $\ker D \hatV \big{|}_{(\bme_0,0)} = T_{(\bme_0,0)} \Slice_{\CC,\bme_0}$, and hence
$T_{(\bme_0,0)} \Slice_{\CC,\NN}^\perp = T_{(\bme_0,0)} \Slice_{\CC,\bme_0}$
as claimed. 
\end{proof} 
Taking lemma \ref{lem:TSlice} into account, we can now complete the analysis
of the shadow gauge condition. 
\begin{prop} \label{prop:A4} 
There is a neighborhoor $\UU$ of $(\bme_0, 0) \in \CC$ such that
  for $(\rme,\rKtr) \in \UU$, there is a unique $\phi \in \DD$, so that 
$$
(\phi^* \rme, \phi^* \rKtr) \in \Slice_{\CC,\NN}^\perp .
$$
In particular, for $(\rme,\rKtr) \in \UU$, there are unique $\phi \in \DD$,
$\bme \in \NN$, such that $(\phi^* \rme, \phi^* \rKtr) \in
\Slice_{\CC,\bme}^\perp$, i.e. such that $(\phi^* \rme, \phi^* \rKtr, \bme)$
satisfy the shadow gauge condition. 
\end{prop} 
\begin{proof} 
We show that the map 
$
L: \DD \times \Slice_{\CC,\NN}^\perp \to \CC ,
$
defined by 
$$
L(\phi,\rme,\rKtr) = (\phi^* \rme, \phi^* \rKtr)
$$
is a diffeomorphism locally at $(\Id, (\bme_0, 0))$, where $\Id$ denotes
the identity in $\DD$. 
Write a general element in $T_{\Id} \DD$ as $X$, and a general element
of $T_{(\bme_0,0)} \Slice_{\CC,\NN}^\perp$ as $(u^\TT,v^\TT)$. We have 
$$
D L\big{|}_{(\Id,(\bme_0,0))} = ( u^\TT + \Lie_X \bme_0, v^\TT) .
$$
Recalling that $T_{(\bme_0,0)} \CC$ is spanned by tensors of the form 
$(u^\TT + \Lie_X \bme_0$, $v^\TT)$, we see that 
the Frechet derivative 
$DL$ is an isomorphism 
$$
DL\big{|}_{(\Id,(\bme_0,0))} : T_{\Id} \DD \times 
T_{(\bme_0,0)}\Slice_{\CC,\NN}^\perp  \to T_{(\bme_0,0)}\CC,
$$
and the proposition follows. 
\end{proof} 
Proposition \ref{prop:A4} establishes the validity of Lemma 
\ref{lem:shadow-exist-met}. 


\providecommand{\bysame}{\leavevmode\hbox to3em{\hrulefill}\thinspace}
\providecommand{\MR}{\relax\ifhmode\unskip\space\fi MR }
\providecommand{\MRhref}[2]{%
  \href{http://www.ams.org/mathscinet-getitem?mr=#1}{#2}
}
\providecommand{\href}[2]{#2}

\end{document}